\documentclass[journal]{IEEEtran}
\usepackage{amsmath}
\usepackage{mathtools,amsfonts,mathrsfs,amssymb,amsthm}
\usepackage[pdftex]{graphicx}
\usepackage{float}
\usepackage{color}
\usepackage{cite}
\usepackage{acronym}
\usepackage{siunitx}
\usepackage{subfigure}
\usepackage{wasysym}
\usepackage[normalem]{ulem}
\usepackage[plain]{algorithm}
\usepackage{algorithmicx}
\usepackage[dvipsnames]{xcolor}
\usepackage{xspace}
\usepackage{caption}
\usepackage{scalerel}
\usepackage{stfloats}
\usepackage{tikz}
\usetikzlibrary{shapes.misc}
\usetikzlibrary{shapes}
\usetikzlibrary{shapes.geometric}
\usetikzlibrary{positioning}
\usetikzlibrary{plotmarks}
\usetikzlibrary{calc}

\usepackage{placeins}

\usepackage{booktabs}

\usepackage{colortbl}

\usepackage{leftidx} 

\definecolor{cblue}{rgb}{0.466666666666667,0.925490196078431,0.956862745098039}
\definecolor{lightGreen}{rgb}{0.5373, 0.9333, 0.5843}
\definecolor{yellah}{rgb}{1.0000, 0.9098, 0.4902}

\newcommand{\customlabel}[2]{
\protected@write \@auxout {}{\string \newlabel {#1}{{#2}{}}}}

\restylefloat{figure}
\newcounter{mytempeqncnt}
\newtheorem{proposition}{Proposition}

\newtheorem{definition}{Definition}

\newcommand{\labels}[1]{\mathcal{L}(#1)}
\newcommand{\F}[1]{\mathcal{F}(#1)}
\newcommand{\minusQuad}{\kern-1em}
\def\Xspace{\ensuremath{\mathbb{X}}\xspace}
\def\Lspace{\ensuremath{\mathbb{L}}\xspace}
\def\Nspace{\ensuremath{\mathbb{N}}\xspace}
\def\Bspace{\ensuremath{\mathbb{B}_+}\xspace}
\def\Tspace{\ensuremath{\mathbb{T}_+}\xspace}

\def\pibold{\ensuremath{\boldsymbol{\pi}}\xspace}

\def\fbold{\ensuremath{\mathbf{f}}\xspace}

\def\Xbold{\ensuremath{\mathbf{X}}\xspace}
\def\xbold{\ensuremath{\mathbf{x}}\xspace}

\def\B{\ensuremath{\mathrm{B}}\xspace}
\def\S{\ensuremath{\mathrm{S}}\xspace}

\def\T{\ensuremath{\mathrm{T}}\xspace}
\def\D{\ensuremath{\mathrm{D}}\xspace}

\def\hex{\ensuremath{{(h)}}\xspace}
\def\htex{\ensuremath{{(h,t)}}\xspace}
\def\tab{\hspace{12pt}}

\def\pS{\ensuremath{p_\S(x,\ell)}\xspace}
\def\qS{\ensuremath{q_\S(x,\ell)}\xspace}
\def\pD{\ensuremath{p_\D(x,\ell)}\xspace}
\def\qD{\ensuremath{q_\D(x,\ell)}\xspace}
\def\splus{\ensuremath{{\S,+}}\xspace}
\def\bplus{\ensuremath{{\B,+}}\xspace}
\def\dplus{\ensuremath{{\D,+}}\xspace}
\def\tplus{\ensuremath{{\T,+}}\xspace}

\acrodef{rfs}[RFS]{\textit{random finite set}}
\acrodef{glmb}[GLMB]{Generalized Labeled Multi-Bernoulli}
\acrodef{lmb}[LMB]{Labeled Multi-Bernoulli}
\acrodef{fisst}[FISST]{Finite Set Statistics}
\acrodef{glmb}[GLMB]{Generalized Labeled Multi-Bernoulli}
\acrodef{phd}[PHD]{Probability Hypothesis Density}
\acrodef{kld}[KLD]{Kullback-Leibler divergence}
\acrodef{resp}[resp.]{respectfully}
\acrodef{mcmc}[MCMC]{Markov Chain Monte Carlo}
\acrodef{ssa}[SSA]{Space Sitational Awareness}
\acrodef{mht}[MHT]{Multiple Hypotheses Tracking}
\acrodef{jpda}[JPDA]{Joint Probabilistic Data Association}
\acrodef{ospa}[OSPA]{Optimal Sub-Pattern Assignment}
\acrodef{pdf}[pdf]{probability density function}
\acrodef{wrt}[w.r.t.]{with respect to}
\begin{document}

\title{A Generalized Labeled Multi-Bernoulli Filter with Object Spawning}
\author{ Daniel~S.~Bryant, Ba-Tuong~Vo, Ba-Ngu~Vo, and Brandon~A.~Jones 
\thanks{
D. S. Bryant is with the Department of Aerospace Engineering Sciences,
University of Colorado Boulder, Boulder, CO 80309 USA (e-mail:
daniel.bryant@colorado.edu). } \thanks{
B.-T. Vo and B.-N. Vo are with the Department of Electrical and Computer
Engineering, Curtin University, Bentley, WA 6102, Australia (e-mail:
ba-tuong.vo@curtin.edu.au; ba-ngu.vo@curtin.edu.au). } \thanks{
B. A. Jones is with the Department of Aerospace Engineering and Engineering
Mechanics, The University of Texas Austin, Austin, TX, 78712 USA (e-mail:
brandon.jones@utexas.edu). } }
\maketitle

\begin{abstract}
Previous labeled random finite set filter developments use a motion model that
only accounts for survival and birth. While such a model provides the means for a
multi-object tracking filter such as the \ac{glmb}
filter to capture object births and deaths in a wide variety of
applications, it lacks the capability to capture spawned tracks and their
lineages. In this paper, we propose a new \ac{glmb} based filter that
formally incorporates spawning, in addition to birth. This formulation
enables the joint estimation of a spawned object's state and information
regarding its lineage. Simulations results demonstrate the efficacy of the
proposed formulation.
\end{abstract}



\acresetall 


\section{Introduction}

Multi-object tracking is concerned with estimating the number of objects and their trajectories in the presence of object appearance/disappearance, clutter, and uncertainties in detection, measurements, and data associations. The field covers a wide variety of applications which include aviation \cite{1996_MTT_aviation}, astrodynamics, \cite{jonesBryantFusion,hecker1995expert,schumSADIE}, defense \cite{blackman1999modern}, robotics \cite{MullaneVo2011_RFS_SLAM,lee2013slam}, and cell biology \cite{Meijering2012183_[cell_tracking]}, \cite{reza2012visual}. Three of the most prominent approaches to multi-object filtering are \ac{mht} \cite{reid1979mht,kurien1990issues,BlackmanMHT2004,mallick2012}, \ac{jpda} \cite{bar-shalom2011_yellow}, and \acp{rfs} \cite{mahler2007statistical,mahler2014advances}.

In a multi-object system, new objects appear either from spontaneous birth or spawning from existing objects. While many existing multi-object tracking approaches can accommodate spontaneous birth, so far only the \ac{rfs} approach offers a principled treatment of spawning in terms of modeling and estimation \cite{mahler2003}, \cite{mahler2007statistical}. Central to the \ac{rfs} approach is the \textit{multi-object Bayes recursion} \cite{mahler2007statistical} which recursively propagates the multi-object posterior forward in time. The first order approximation more commonly known as the \ac{phd} filter accounts for both birth and spawning \cite{mahler2003}. However, its generalization, the Cardinalized PHD (CPHD) filter, was derived without a spawning model \cite{mahler2007phd}. The CPHD filter with spawning is generally intractable and approximations were derived in \cite{Lundgren2013}, \cite{Bryant_CPHD_Spawn}.

Information on lineage or ancestry is an important aspect of tracking multiple spawning objects. For example, in (biological) cell tracking, information on a cell's lineage is important to the analyis of cell behavior \cite{Li2008,becker2012morphology_[cell_tracking],Meijering2012183_[cell_tracking],Baker20140386_[cell_tracking]}. For space situational awareness, information on the ancestry of debris is important to the analysis of fragmentation events \cite{BryantJones_Napa_2016,schumacher1996tracking,DemarsBreakup2015}; moreover, country of origin and launch site, information required to add a space object to the United States Strategic Command (USSTRATCOM) catalog \cite{spaceTrackFAQ}, can be derived from ancestry information. Even with spawning models, the PHD/CPHD filters \cite{mahler2003}, \cite{Lundgren2013}, \cite{Bryant_CPHD_Spawn} only provide estimates of spawned objects' states, but no information on their ancestries. Further, in applications where ancestry information is not required, it is also possible to estimate spawned objects using \ac{rfs}-based multi-object filters with measurement-driven (spontaneous) birth models \cite{ristic2012adaptive}, \cite{mahlermaroulas2013}. Hence, a complete treatment of modeling and estimation for spawning objects should address the issue of ancestry.

\textit{Labeled \acp{rfs}} enable ancestry information to be incorporated into the modeling and estimation of spawning objects. Approximate multi-object Bayes filters such as the \ac{phd} \cite{mahler2003}, CPHD \cite{mahler2007phd}, and multi-Bernoulli \cite{mahler2007statistical, Vo2010jointDetect, vo2009cardinality} filters were not formulated to estimate tracks, without which the ancestries of the objects are, conceptually, not traceable. On the other hand, labeled \acp{rfs} provide the means for identifying and estimating individual object tracks \cite{Vo_dGLMB_ConjPrior_2013}, thereby making it possible, conceptually, to trace their ancestors. Furthermore, as demonstrated in this paper, the labels used to identify individual tracks can also be encoded with ancestry information, which can be assimilated by \ac{rfs} spawning models, and subsequently inferred from the labels obtained using labeled \ac{rfs} estimation techniques.

Under the labeled \ac{rfs} formulation, the multi-object Bayes recursion (without spawning) admits an analytic solution known as the \ac{glmb} filter \cite{Vo_dGLMB_ConjPrior_2013,Vo_dGLMB_Sequel_2014}, which can be implemented with linear complexity in the number of measurements and quadratic in the number of hypothesized tracks \cite{Vo_2016_fast_dGLMB}. This on-line multi-object tracker is based on the \ac{glmb} family of conjugate priors that enjoys a number of nice analytical properties, e.g., the void probability functional--a necessary and sufficient statistic--of a \ac{glmb}, the Cauchy-Schwarz divergence between two \acp{glmb} \cite{beard2017}, the $L_{1}$-distance between a \ac{glmb} and its truncation, can all be computed in closed form \cite{Vo_dGLMB_Sequel_2014}. Of direct relevance to this work is the fact that the \ac{glmb} family is flexible enough to approximate any labeled \ac{rfs} density with matching intensity function and cardinality distribution \cite{Papi2015}.

In this paper, we propose a new \ac{glmb} based filter that formally incorporates spawning, in addition to birth.\ Using labeled \acp{rfs} we encode ancestry information into the labels of individual object states and propose a labeled \ac{rfs} spawning model. When a track is instantiated by spontaneous birth, its label contains information pertaining to when an object is born and from which birth region \cite{Vo_dGLMB_ConjPrior_2013}. Similarly, for a track instantiated by spawning, its label contains information pertaining to when and from which parent it originated. Under such a spawning model, the multi-object prediction and filtering densities are no longer \acp{glmb}, even if the initial prior is a \ac{glmb}. To derive a tractable filter, following \cite{Vo_2016_fast_dGLMB} we combine the prediction and update into a single step and approximate the labeled multi-object filtering density by a \ac{glmb} with matching first moment and cardinality using the technique in \cite{Papi2015}. The result is a recursion that propagates the \ac{glmb} approximation of the labeled multi-object filtering density, from which the states of the spawned objects and their labels, hence lineage, can be jointly inferred.

The remainder of this paper is arranged as follows: In Section~\ref{sec:background} background is provided on the \ac{glmb} filter, its use in approximation of general labeled \ac{rfs} densities, and its efficient implementation. In Section~\ref{sec:newStuff} the derivation, approximation, and joint prediction and update of object spawning inclusive \ac{glmb} densities is developed. Simulation results are presented in Section~\ref{sec:simulation} and in Section~\ref{sec:conclusion} concluding remarks are given.

\section{Background\label{sec:background}}

This section provides background on \ac{glmb} filter implementation pertinent to the formulation of our multi-object filtering problem.

	\subsection{Labeled Random Finite Sets}

	Instead of a single system state $x$ in the state space $\mathbb{X}\xspace$, in a multi-object system we consider a finite set $X\subset \mathbb{X} \xspace$ as the \textit{multi-object state}. Further, in a Bayesian framework the multi-object state is modeled as an \ac{rfs}, i.e., a finite-set-valued random variable \cite{mahler2003}. An \ac{rfs}, also known as a simple finite point process, consists of a random number of points that are, themselves, random and unordered. An \ac{rfs} can be described by the  \textit{multi-object density--}defined to be\textit{\ }the set derivative of its belief functional \cite{mahler2007statistical}--shown to be equivalent to a probability density in \cite{Vo_SMCPHD}.
	
	A \emph{labeled \ac{rfs}} is a marked simple finite point process with state space $\mathbb{X}\xspace$ and discrete mark space $\mathbb{L}$, such that each realization has distinct marks \cite{Vo_dGLMB_ConjPrior_2013}, \cite{Vo_dGLMB_Sequel_2014}. The distinct marks or labels provide the means to identify trajectories or tracks of individual objects since a trajectory is a time-sequence of states with the same label. Let $\mathcal{L}:\mathbb{X}\xspace\times \mathbb{L}\xspace\rightarrow \mathbb{L}\xspace$ be the projection $\mathcal{L}((x,\ell ))=\ell $, then the labels of realization $\mathbf{X}\xspace\subset \mathbb{X}\xspace\times \mathbb{L}\xspace$ are then $\mathcal{L}(\mathbf{X}\xspace)=\left\{ \mathcal{L}(\mathbf{x}\xspace):\mathbf{x}\xspace\in \mathbf{X}\xspace\right\} $. The realization $\mathbf{X}$ is said to have \textit{distinct labels} if and only if it has the the same cardinality as its labels $\mathcal{L}(\mathbf{X}\xspace)$. This concept is compactly formulated by the \textit{distinct label indicator} defined by \cite{Vo_dGLMB_ConjPrior_2013}, \cite{Vo_dGLMB_Sequel_2014} 
	\begin{equation*}
	\Delta (\mathbf{X}\xspace)=\delta _{|\mathbf{X}\xspace|}\left( |\mathcal{L}(%
	\mathbf{X}\xspace)|\right) ,
	\end{equation*}%
	where $|X|$ denotes the cardinality of a finite set $X$, and
	\begin{equation*}
	\delta _{Y}(X)\triangleq 
	\begin{cases}
	1,\text{ if }X=Y, \\ 
	0,\text{ otherwise,}%
	\end{cases}%
	\end{equation*}%
	denotes a generalization of the Kroneker delta that takes arbitrary arguments such as integers, sets, vectors etc.
	
	Throughout this paper we adhere to the convention that lower case letters represent single-object states, e.g., $x,\mathbf{x}\xspace$, while upper case letters represent multi-object states, e.g., $X,\mathbf{X}\xspace$. Bold symbols represent labeled states and their distributions/statistics, e.g., $\mathbf{x}\xspace,\mathbf{X}\xspace,\boldsymbol{\pi }\xspace$, etc., to distinguish them from unlabeled ones. Blackboard letters represent spaces, e.g., $\mathbb{X}\xspace,\mathbb{Z}\xspace,\mathbb{L}\xspace,\mathbb{N}\xspace$. We use the standard inner product notation $\langle f,g\rangle \triangleq \int f(x)g(x)dx$, the multi-object exponential notation $h^{X}\triangleq \prod_{x\in X}h(x)$, where $h$ is a real-valued function and $h^{\emptyset }=1$ by convention, and the inclusion function notation, a generalization of the indicator function, by 
	\begin{equation*}
	1_{Y}(X)\triangleq 
	\begin{cases}
	1,\text{ if }X\subseteq Y, \\ 
	0,\text{ otherwise.}%
	\end{cases}%
	\end{equation*}%
	Additionally, the list of variables $X_{m},X_{m+1},...,X_{n}$ is abbreviated as $X_{m:n}$ and where it is convenient we let the symbol $+$ denote the time index at the \textit{next time} and its absence denote the time index at the \textit{current time}, e.g., the state $x_{k}$ at the current time and the state $x_{k+1}$ at the next time can equivalently be denoted as $x$ and $x_{+}$, respectively.

	\subsection{Generalized Labeled Multi-Bernoulli}

	A \ac{glmb} is a labeled \ac{rfs} with state space $\mathbb{X}$\xspace and label space $\mathbb{L}$\xspace distributed according to \cite{Vo_dGLMB_ConjPrior_2013,Vo_dGLMB_Sequel_2014} 
	\begin{equation} \label{eq:dglmbDist}
	\pibold(\Xbold) = \Delta(\Xbold)\sum_{(I,\xi)\in\F{\Lspace}\times\Xi}w^{(I,\xi)}\delta_I(\labels{\Xbold})\left[p^{(\xi)}\right]^\Xbold
	\end{equation}
	where $\Xi$ is a given discrete space, each $p^{(\xi )}(\cdot ,\ell )$ is a (single-object) probability density on \Xspace (i.e., $\int p^{(\xi )}(x,\ell )dx=1$ with each $x\in \Xspace$ denoting a single-object state and each $\ell\in\Lspace$ denoting a distinct label), and each $w^{(I,\xi )}$ is non-negative such that 
	\begin{equation} \label{eq:wsumto1}
	\sum_{I\in\F{\Lspace}}\sum_{\xi\in\Xi}w^{(I,\xi)}(L)=1.
	\end{equation}
	Each \ac{glmb} density component $(I,\xi)$ in \eqref{eq:dglmbDist} consists of a weight $w^{(I,\xi)}$ that depends solely on the labels of the multi-object state \Xbold and the multi-object exponential $\left[ p^{(\xi )}\right]^{\Xbold}$, which is a product of single-object probability densities.
	
	Also relevant to this work is the \ac{lmb}. An \ac{lmb} \Xbold defined on $\Xspace\times\Lspace$ is an \ac{rfs} with parameter set $\{(r^{(\zeta)},p^{(\zeta)}):\zeta\in\Psi\}$ distributed according to \cite{Vo_dGLMB_ConjPrior_2013}
	\begin{equation}
	\pibold(\Xbold) = \Delta(\Xbold)1_{\alpha(\Psi)}(\labels{\Xbold})\left[\Phi(\Xbold,\cdot)\right]^\Psi
	\end{equation}
	where \mbox{$\alpha : \Psi\rightarrow\Lspace$} is a 1-1 mapping (usually an identity mapping) and 
	\begin{align}
	\Phi(\Xbold,\zeta) = \sum_{(x,\ell)\in\Xbold}&\delta_{\alpha(\zeta)}(\ell)r^{(\zeta)}p^{(\zeta)}(x)\nonumber
	\\
	&+\left(1-1_{\labels{\Xbold}}(\alpha(\zeta))\right)\left(1-r^{(\zeta)}\right).
	\end{align}

	\subsection{Multi-object Bayes Filter}
	Using the convention detailed in \cite{Vo_dGLMB_ConjPrior_2013}, a label $\ell =(t,i)$ in the space \Lspace of labels at the current time $k$ is an ordered pair, where the first term $t\leq k$ denotes time of birth, and the second term $i\in \Nspace$ is a unique index distinguishing objects born at the same time. Birth labels at the next time belong to the space $\Bspace=\{(k+1,i):i\in \Nspace\}$, hence $\Lspace\cap\Bspace=\emptyset $ and the label space at the next time is $\Lspace_{+}=\Lspace\cup \Bspace$.

	The history $\Xbold_{0:k}$ of labeled multi-object states contains the set of all trajectories up to time $k$. All information on the set of trajectories conditioned on the observation history $Z_{1:k}$, is captured in the \emph{multi-object posterior density }$\pibold_{0:k}(\mathbf{\cdot }|Z_{1:k}\mathbf{)}$, which incorporates the evolution of the multi-object state via the \textit{multi-object transition density}, as well as the observed data via the \textit{multi-object likelihood} \cite{mahler2007statistical}.

		\subsubsection{Multi-Object Transition Model}

		The \textit{multi-object transition density} $\fbold_{+}(\mathbf{\cdot }| \Xbold)$ models the evolution of a given multi-object state \Xbold to the next time and encapsulates all information pertaining to loss of objects via thinning, movement of surviving objects via Markov shifts, and appearance of new objects via superposition.

		Given a single-object state $\xbold\in\Xbold$ at the current time, an object either survives to the next time with probability \pS and moves to a new state $(x_{+},\ell_{+}) $ with probability density $f_\splus(x_+ | x,\ell)\delta_\ell(\ell_+)$, or dies with probability $\qS = 1-\pS$. Assuming that, conditional on \Xbold, the transition of kinematic states $\xbold\in\Xbold$ are mutually independent, we model the set $\Xbold_\splus$ of surviving objects at the next time as a conditional \ac{lmb} \ac{rfs} distributed according to \cite{Vo_dGLMB_ConjPrior_2013} 
		\begin{align} \label{eq:survKern}
		\fbold_\splus(&\Xbold_\splus|\Xbold) = \nonumber
		\\
		&\Delta(\Xbold)\Delta(\Xbold_\splus)1_{\labels{\Xbold}}(\labels{\Xbold_\splus})\left[\Phi_\splus(\Xbold_\splus|\cdot)\right]^\Xbold
		\end{align}
		where 
		\begin{align*}
		\Phi_\splus(\Xbold_\splus|x,\ell) = \!\!\!\!\sum_{(x_+,\ell_+)\in\Xbold_\splus} \!\!\!\! &\delta_\ell(\ell_+)\pS f_\splus(x_+|x,\ell) \nonumber
		\\
		&+ \left[1-1_{\labels{\Xbold_\splus}}(\ell)\right]\qS.
		\end{align*}
		
		A new object with state $(x_{+},\ell _{+})$ appears at the next time with probability $r_\bplus(\ell_+)$ and probability density $p_\bplus(x_+,\ell_+)$, or does not with probability $1-r_\bplus(\ell_+)$. Modeling object birth as an \ac{lmb} \ac{rfs}, the set $\Xbold_\bplus$ of new objects born at the next time is distributed according to \cite{Vo_dGLMB_ConjPrior_2013}
		\begin{equation} \label{eq:birthKern}
		\fbold_\bplus(\Xbold_\bplus)=\Delta(\Xbold_\bplus)w_\bplus(\labels{\Xbold_\bplus})\left[p_\bplus\right]^{\Xbold_\bplus}
		\end{equation}
		where $w_\bplus(L)=1_{\Bspace(L)}\left[1-r_\bplus\right]^{\Bspace-L}\left[r_\bplus\right]^L$.

		The multi-object state at the next time is the superposition of birth and surviving objects, i.e., $\Xbold_+=\Xbold_\splus \cup \Xbold_\bplus$, and since $\Lspace\cap\Bspace=\emptyset $, \textit{labeled} birth and surviving objects are mutually independent. Thus, the multi-object transition kernel ultimately reduces to the product of birth and survival transition densities \cite{Vo_dGLMB_ConjPrior_2013} 
		\begin{equation} \label{eq:sbKern}
		\fbold_+(\Xbold_+|\Xbold) = \fbold_\splus(\Xbold_+\cap(\Xspace\times\Lspace)|\Xbold)\fbold_\bplus(\Xbold_+-(\Xspace\times\Lspace)),
		\end{equation}
		where $\Xbold_+\cap(\Xspace\times\Lspace)$ is the subset of $\Xbold_+$ consisting of surviving objects.

		\subsubsection{\mbox{Multi-object Measurement Model}\label{sec:measModel}}%
		The multi-object likelihood is a multi-object density $g(\cdot |\mathbf{X}\xspace)$ that models the \textit{multi-object observation} generated by a given multi-object state $\mathbf{X}$\xspace, and	encapsulates all information pertaining to missed detections via thinning, detections (observations of detected objects) via Markov shifts and clutter (false observations) via superposition.

		The multi-object observation $Z=\left\{ z_{1},\ldots ,z_{|Z|}\right\} $ is the superposition of detections and clutter. Each state $(x,\ell )\in \mathbf{X}\xspace$ is either detected with probability \pD and generates an observation $z\in Z$ with likelihood $g(z|x,\ell )$ or missed with probability $\qD = 1-\pD$. The multi-object likelihood is given by \cite{Vo_dGLMB_ConjPrior_2013,Vo_dGLMB_Sequel_2014} 
		\begin{equation} \label{eq:multLikeGLMB}
		g(Z|\Xbold) \propto \sum_{\theta \in\Theta(\mathcal{L}(\Xbold))}\prod_{(x,\ell)\in\Xbold}\psi^{(\theta(\ell))}(x,\ell|Z)
		\end{equation} 
		where 
		\begin{align*}
		\psi^{(j)}\left(x,\ell|\left\{Z_{1:|Z|}\right\}\right) = & \;\delta_0(j)\qD \nonumber
		\\
		&\hspace{5pt}+ (1-\delta_0(j))\frac{\pD g(z_j|x,\ell)}{\kappa(z_j)},
		\end{align*}
		$\kappa (\cdot )$ is the intensity of Poisson clutter, and $\Theta (L)$ denotes the space of mappings $\theta\!\!\!:\!\!\!L\!\!\rightarrow\!\! \{0\!\!:\!\!|Z|\}$ that are 1-1 when restricting the range to the positive integers, i.e., $\theta(i)=\theta (j)>0\text{ implies }i=j$.

		\subsubsection{Multi-object Bayes recursion}
		While the multi-object posterior density can be approximated by Markov Chain Monte Carlo \cite{Vu2014,Craciun2015}, these techniques are still expensive and not suitable for on-line applications. For real-time tracking, a more tractable alternative is the marginal $\boldsymbol{\pi }(\mathbf{\cdot)\triangleq }\boldsymbol{\pi }_{k}(\mathbf{\cdot}|Z_{1:k}\mathbf{)}$ called the \emph{multi-object filtering density}, which can be recursively propagated by the \emph{multi-object Bayes filter} \cite{mahler2003}, \cite{mahler2007statistical}
		\begin{equation} \label{Chap_Kolmo_LRFS}
		\pibold_+(\Xbold_+)\propto g_+(Z_+|\Xbold_+)\int \fbold_+(\Xbold_+|\Xbold)\pibold(\Xbold)\delta\Xbold,
		\end{equation}
		where the integral is a \emph{set integral} defined for any function $\fbold : \mathcal{F}(\Xspace\times\Lspace_k)\rightarrow \mathbb{R}$ by 
		\begin{equation*}
		\int \fbold(\Xbold)\delta\Xbold = \sum_{i=0}^\infty \frac{1}{i!}\int\fbold(\{\xbold_1,\ldots,\xbold_i\})d(\xbold_1,\ldots,\xbold_i).
		\end{equation*}
		Note that Bayes optimal multi-object estimators can be formulated by minimizing the Bayes risk, e.g., the marginal multi-object estimator \cite{mahler2007statistical}.
		
		An analytic solution to the labeled Bayes multi-object filter (\ref{Chap_Kolmo_LRFS}), known as the \emph{Generalized Labeled Multi-Bernoulli} (GLMB) filter, was derived in \cite{Vo_dGLMB_ConjPrior_2013}, while a particle approximation was implemented in \cite{PapiKim2015} using the generic multi-object particle algorithm proposed in \cite{Vo_SMCPHD}.

	\subsection{Fast GLMB Filter Implementation \label{sec:fastImplBack}}
	The first \ac{glmb} filter implementation consists of prediction and update stages, each requiring independent truncations of \ac{glmb} densities \cite{Vo_dGLMB_ConjPrior_2013,Vo_dGLMB_Sequel_2014}. Alternatively, a substantially more efficient implementation of the \ac{glmb} filter \cite{Vo_2016_fast_dGLMB}, hereafter referred to as the \textit{fast} \ac{glmb} implementation, employs a single joint prediction/update stage requiring only one truncation procedure. This work employs the fast \ac{glmb} implementation, thus for convenience, we introduce pertinent expressions and conventions for \ac{glmb} joint prediction/update and formulation of the \ac{glmb} truncation problem originally presented in \cite{Vo_2016_fast_dGLMB}. We expand on this material in Section~\ref{sec:Implement} to incorporate spawning.

	Given the \ac{glmb} filtering density \eqref{eq:dglmbDist} at the current time, the \ac{glmb} filtering density at the next time is given by \cite{Vo_2016_fast_dGLMB}\footnote{In the interest of simplifying notation, note that $\sum\limits_{(I,\xi )\in\mathcal{F}(\mathbb{L}\xspace)\times\Xi}a^{(I,\xi)}=\sum\limits_{I,\xi}a^{(I,\xi )}$ when the definitions $I\in \mathcal{F}(\mathbb{L}\xspace)$ and $\xi \in \Xi $ are provided.} 
	\begin{align} \label{eq:fastJointDensity}
	\pibold_{+}(\Xbold_+) \propto &\Delta(\Xbold_+)\sum_{I,\xi,I_+,\theta_+}w^{(I,\xi)}w^{(I,\xi,I_+,\theta_+)}(Z_+)\nonumber
	\\
	&\times\delta_{I_+}(\labels{\Xbold_+})\left[p_+^{(\xi,\theta_+)}(\cdot | Z_+)\right]^{\Xbold_+}
	\end{align}
	where $I\in\F{\Lspace}$, $\xi\in\Xi$, $I_+\in\F{\Lspace_+}$, $\theta_+\in\Theta_+(I_+)$, and	
	\begin{align}
	\hspace*{-8pt}w^{(I,\xi,I_+,\theta_+)}(Z_+) &= \!\left[1-\bar{p}_\S^{(\xi)}\right]^{I-I_+}\left[\bar{p}_\S^{(\xi)}\right]^{I\cap I_+}\nonumber
	\\
	&\hspace{10pt}\times\left[1-r_\bplus\right]^{\Bspace-I_+}\!\left[r_\bplus\right]^{\Bspace\cap I_+}\nonumber
	\\
	&\hspace{10pt}\times\left[\bar{\psi}_+^{(\xi,\theta_+)}(\cdot |Z_+)\right]^{I_+},
	\\
	\hspace*{-8pt}\bar{p}_\S^{(\xi)}(\ell) &= \left\langle p_\S(\cdot,\ell),p^{(\xi)}(\cdot,\ell) \right\rangle,\label{eq:barP_S}
	\\
	\hspace*{-8pt}\bar{\psi}_+^{(\xi,\theta_+)}\!(\ell_+|Z_+\!) \!&=\! \left\langle \bar{p}_+^{(\xi)}(\cdot,\ell_+),\psi_+^{(\theta_+(\ell_+))}\!(\cdot,\ell_+|Z_+)\! \right\rangle,
	\\
	\hspace*{-8pt}p_+^{(\xi,\theta_+)}\!(x_+,\!\ell_+|Z_+\!) \!&=\! \frac{\bar{p}_+^{(\xi)}\!(x_+,\!\ell_+\!)\psi_+^{(\theta_+(\ell_+))}\!(x_+,\!\ell_+|Z_+\!)}{\bar{\psi}_+^{(\xi,\theta_+)}(\ell_+|Z_+)},
	\\
	\hspace*{-8pt}\bar{p}_+^{(\xi)}(x_+,\ell_+) &=1_{\Bspace}(\ell_+)p_\bplus(x_+,\ell_+) + 1_\Lspace(\ell_+) \nonumber
	\\
	&\hspace{-10pt}\times\frac{\left\langle p_\S(\cdot,\ell_+)f_\splus(x_+|\cdot,\ell_+),p^{(\xi)}(\cdot,\ell_+)\right\rangle}{\bar{p}_\S^{(\xi)}(\ell_+)}. \label{eq:14wouldNotFit}
	\end{align}
	
	Though \eqref{eq:fastJointDensity} is not strictly \ac{glmb}, it does take on \ac{glmb} form when rewritten as a sum over $I_{+},\xi ,\theta _{+}$ with weights \cite{Vo_2016_fast_dGLMB} 
	\begin{equation}
	w_{_{+}}^{(I_{+},\xi ,\theta _{+})}(Z_{+})\propto \sum_{I}w^{(I,\xi
	)}w^{(I,\xi ,I_{+},\theta _{+})}(Z_{+}).  \label{eq:makeDGLMB}
	\end{equation}

	Efficient implementation of the GLMB recursion \eqref{eq:fastJointDensity} is achieved by propagating only the components with significant $w^{(I,\xi,I_{+},\theta _{+})}(Z_{+})$ through time, i.e., for each component $(I,\xi) $ from the \ac{glmb} density at the current time and a multi-object observation $Z_{+}$ at the next time, the set of pairs $(I_{+},\theta_{+})\in \mathcal{F}(\mathbb{L}\xspace)\times \Theta _{+}(I_{+})$ with significant $w^{(I,\xi ,I_{+},\theta _{+})}(Z_{+})$ are retained while the rest are discarded. The truncation procedure is described as follows.

	Consider a fixed component $(I,\xi )$, and enumerate $Z_{+}=\{z_{1:|Z_{+}|}\} $, $\mathbb{B}_{+}\xspace=\{\ell _{1:K}\}$, and $I=\{\ell _{K+1:P}\}$. For each pair $(I_{+},\theta_{+})\in \mathcal{F}(\mathbb{L}\xspace)\times\Theta _{+}(I_{+})$, an equivalent $P$-dimensional vector representation $\gamma =(\gamma _{1:P})\in \{-1:|Z_{+}|\xspace\}^{P}$ is defined as 
	\begin{equation} \label{eq:gamma}
	\gamma _{i}=
	\begin{cases}
	\theta _{+}(\ell _{i}), & \text{ if }\ell _{i}\in I_{+}, \\ 
	-1, & \text{ otherwise.}
	\end{cases}
	\end{equation}
	Note that $\gamma $ inherits the positive 1-1 property from $\theta _{+}$, and that $I_{+}$ and $\theta _{+}:I_{+}\rightarrow \{0:|Z_{+}|\xspace\}$ can be recovered by 
	\begin{equation}
	I_{+}=\left\{ \ell _{i}\in \mathbb{B}_{+}\xspace\cup \ I:\gamma _{i}\geq 0\right\} ,\hspace{1cm}\theta _{+}(\ell _{i})=\gamma _{i}.
	\end{equation}
	Assuming that, for all $i\in \{1:P\}$, $\bar{p}_\S^{(\xi)}(\ell _{i})\in (0,1)$ and $\bar{p}_\dplus^{(\xi)}\xspace(\ell _{i})\triangleq\left\langle p_\dplus(\cdot ,\ell_{i}),\bar{p}_{+}^{(\xi )}(\cdot ,\ell _{i})\right\rangle \in (0,1)$ , we define for each $j\in \{-1:|Z_{+}|\xspace\}$ 
	\begin{equation} \label{eq:origCostMatrix}
	\hspace*{-5pt}\eta _{i}(j)=
	\begin{cases}
	1-r_\bplus(\ell _{i}), &\!\!\!  1\!\leq\! i\!\leq\! K,j\!<\!0, \\ 
	r_\bplus(\ell _{i})\bar{\psi}_{+}^{(\xi ,j)}(\ell_{i}|Z_{+}), &\!\!\! 1\!\leq\! i\!\leq\! K,j\!\geq\! 0, \\ 
	1-\bar{p}_\S^{(\xi)}(\ell _{i}), &\!\!\! K+1\!\leq\! i\!\leq\! P,j\!<\!0, \\ 
	\bar{p}_\S^{(\xi)}(\ell _{i})\bar{\psi}_{+}^{(\xi,j)}(\ell _{i}|Z_{+}), &\!\!\! K+1\!\leq\! i\!\leq\! P,j\!\geq\! 0,%
	\end{cases}
	\end{equation}
	where $\bar{\psi}_{+}^{(\xi ,j)}(\ell _{i}|Z_{+})=\left\langle \bar{p}_{+}^{(\xi )}(\cdot ,\ell _{i}),\psi _{+}^{(j)}(\cdot ,\ell_{i}|Z_{+})\right\rangle $. Then $w^{(I,\xi,I_{+},\theta_{+})}(Z_{+})=\prod_{i=1}^{P}\eta _{i}(\gamma _{i})$, if the positive 1-1 vectors $\gamma $ are equivalent representations of $(I_{+},\theta _{+})$. Hence generating significant GLMB children components of $(I,\xi )$ translates to generating positive 1-1 vectors with significant weights \cite{Vo_2016_fast_dGLMB}. Methods for obtaining a set of positive 1-1 vectors include:
	\begin{itemize}
	\item solving a ranked assignment problem using Murty's algorithm \cite{Murty_1968}, which finds the $N$ best 1-1 vectors in non-increasing order; and
	\item a more efficient method using the Gibbs sampler to simulate an unordered set of significant positive 1-1 vectors \cite{Vo_2016_fast_dGLMB}.
\end{itemize}

\section{GLMB Filter with Spawning \label{sec:newStuff}}

This section presents a labeled \ac{rfs} spawning model and a tractable multi-object filter for such a spawning model. In Subsection~\ref{sec:dglmbPredSpawn}, we detail the labeled \ac{rfs} spawning model. In Subsection~\ref{sec:dglmbPred} we derive the resulting prediction multi-object density at the next time for a given GLMB at the current time. We discuss the method of GLMB approximation which matches first moment and cardinality, then utilize it in the derivation of the multi-object filtering density at the next time in Subsection~\ref{sec:glmbUpdApprox}. Implementation of the resulting GLMB recursion is detailed in Subsection~\ref{sec:Implement}.

	\subsection{Multi-object Labeled Spawning Model\label{sec:dglmbPredSpawn}}

	To encode ancestry information in the labels, we adhere to the following labeling convention for spawned tracks. An object spawned from a parent with label $\ell $, at time $k+1$, has label $\varsigma =(\ell ,k+1,i)$, where $i$ is an index that distinguishes multiple objects simultaneously spawned by the same parent. As a result, spawn labels consist of an ancestral element, i.e., the parent's label, and a non-ancestral element that distinguishes multiple spawned objects originating simultaneously from the same parent. Hence, given the label space $\mathbb{L}\xspace_{k}$ for objects at the current time, the label space $\mathbb{T}_{k+1}$\xspace for tracks spawned at the next time is given by $\mathbb{T}_{k+1}\xspace=\mathbb{L}\xspace_{k}\times \{k+1\}\times \mathbb{N}\xspace$.

	Hereafter, in an effort to simplify notation, we revert to the convention of letting the symbol $+$ denote the \textquotedblleft next time" index, e.g., the label space $\mathbb{L}\xspace_{k}$ at the current time becomes $\mathbb{L}\xspace$, and the label space $\mathbb{L}\xspace_{k+1}$ at the next time becomes $\mathbb{L}\xspace_{+}$. Accordingly, we follow the same construction in \cite{Vo_dGLMB_ConjPrior_2013} by letting $\mathbb{T}_{+}$\xspace denote the label space for objects spawned at the next time, then $\mathbb{L}\xspace_{+}=\mathbb{L}\xspace\cup \mathbb{T}_{+}\xspace\cup\mathbb{B}_{+}\xspace$. Note that $\mathbb{L}$\xspace, $\mathbb{T}_{+}$\xspace, and $\mathbb{B}_{+}$\xspace are mutually disjoint, i.e., $\mathbb{L}\xspace\cap \mathbb{T}_{+}\xspace=\mathbb{L}\xspace\cap \mathbb{B}_{+}\xspace=\mathbb{T}_{+}\xspace\cap \mathbb{B}_{+}\xspace=\emptyset $. Hence, we can distinguish surviving, spawn, and birth objects from their labels. Fig.~\ref{fig:spawnlabels}, modeled after \cite[Fig.1]{Vo_dGLMB_Sequel_2014} in the interest of consistency, illustrates label assignment to birth and spawn tracks.

	The elements of a given current multi-object state $\mathbf{X}$ spawn new objects independently of each other. In addition, the set $\mathbf{U}\xspace$ of objects spawned at the next time by a single-object state $\mathbf{x}\triangleq (x,\ell )$, is an \ac{lmb}\footnote{It is possible to derive a \ac{glmb} based spawning model, but an \ac{lmb} is presented for compactness} with parameter set $\{(p_\T(\mathbf{x;}\varsigma ),f_\tplus(\cdot |\mathbf{x;}\varsigma)):\varsigma \in\mathbb{T}_{+}(\mathcal{L}(\mathbf{x)})\}$, where 
	\begin{equation}
	\mathbb{T}_{+}(\ell )\xspace\triangleq \{(\ell ,k+1)\}\times \{1:M_{\ell }\} \label{eq:Tequation}	
	\end{equation}
	is a finite subset set of $\mathbb{T}_{+}\xspace$. In other words, for each LMB component $\varsigma \in \mathbb{T}_{+}(\mathcal{L}(\mathbf{x)})$, the state $\mathbf{x}$ either spawns a state $(x_{+},\varsigma )$ with probability $p_\T(\mathbf{x;}\varsigma )$ and probability density $f_\tplus(x_{+}|\mathbf{x;}\varsigma )$, or it does not with probability $q_\T(\mathbf{x;}\varsigma )=1-p_\T(\mathbf{x;}\varsigma )$. The density of the set of objects spawned from $\mathbf{x}$ can be written as 
	\begin{align} \label{eq:spawnKern}
	\fbold_\tplus(\mathbf{U}\xspace|\mathbf{x})=\Delta (\mathbf{U}\xspace)&1_{\mathbb{T}_{+\!}(\mathcal{L}(\mathbf{x)})}(\mathcal{L}(\mathbf{U}\xspace))\! \nonumber
	\\
	&\hspace{15pt}\times\left[\Phi _\tplus(\mathbf{U}\xspace|\mathbf{x;\cdot })\right]^{\mathbb{T}_{+}(\mathcal{L}(\mathbf{x)})}
	\end{align}
	where
	\begin{align}
	\Phi_\tplus\left( \mathbf{U}\xspace|\mathbf{x;}\varsigma\right)  = \hspace{-8pt}\sum_{(x_{\!+},\ell _{\!+\!})\in\mathbf{U}}&\hspace{-5pt}\delta _{\varsigma }(\ell _{\!+})p_\T(\mathbf{x;}\varsigma)f_\tplus(x_{+}|\mathbf{x;}\varsigma) \nonumber
	\\
	&+[1-1_{\mathcal{L}(\mathbf{U}\xspace)\!}(\varsigma )]q_\T(\mathbf{x;}\varsigma ).
	\label{eq:spawnKern2}
	\end{align}

	Since $\mathbf{X}$ has distinct labels, the LMB label sets $\mathbb{T}_{+}(\mathcal{L}(\mathbf{x)})$ for all $\mathbf{x}\in \mathbf{X}\xspace$ are mutually disjoint, and the set of possible labels spawned from $\mathbf{X}$ is the disjoint union 
	\begin{equation*}
	\mathbb{T}_{+}(\mathcal{L}(\mathbf{X)})=\biguplus_{\mathbf{x}\in \mathbf{X}%
	\xspace}\mathbb{T}_{+}(\mathcal{L}(\mathbf{x)}).
	\end{equation*}
	Note that when a labeled set $\mathbf{V}\xspace$ is spawned from $\mathbf{X}$, $\mathcal{L}(\mathbf{V})\subseteq \mathbb{T}_{+}(\mathcal{L}(\mathbf{X)})$, i.e., $1_{\mathbb{T}_{+}(\mathcal{L}(\mathbf{X)})}(\mathcal{L}(\mathbf{V}))=1$, otherwise $\mathcal{L}(\mathbf{V})\nsubseteq \mathbb{T}_{+}(\mathcal{L}(\mathbf{X)})$, i.e., $1_{\mathbb{T}_{+}(\mathcal{L}(\mathbf{X)})}(\mathcal{L}(\mathbf{V}))=0$. Hence, the inclusion $1_{\mathbb{T}_{+}(\mathcal{L}(\mathbf{X)})}(\mathcal{L}(\mathbf{V}))$ is an indicator of whether $\mathbf{V}$ is spawned by $\mathbf{X}$ or not. Moreover, if $\mathbf{V}\xspace$ is not spawned from $\mathbf{X}$\textbf{, }then the spawning density $\mathbf{f}_\tplus(\mathbf{V}\xspace|\mathbf{X}\xspace)=0$. On the other hand if $\mathbf{V}\xspace$ is spawned from $\mathbf{X}$\xspace, then
	\begin{equation*}
	\mathbf{V}=\biguplus_{\mathbf{x}\in \mathbf{X}\xspace}\mathbf{V}\xspace\cap (%
	\mathbb{X}\times \mathbb{T}_{+}(\mathcal{L}(\mathbf{x)}),
	\end{equation*}
	and since $\mathbf{V}\xspace\cap (\mathbb{X}\times \mathbb{T}_{+}(\mathcal{L}(\mathbf{x)})$ is the set of objects spawned by $\mathbf{x}$, it follow from the \ac{fisst} fundamental convolution theorem \cite{mahler2007statistical}, and arguments presented in \cite{Vo_dGLMB_ConjPrior_2013} that the spawning density 
	\begin{equation*}
	\mathbf{f}_\tplus(\mathbf{V}\xspace|\mathbf{X}\xspace)=\prod_{\mathbf{x}\in\mathbf{X}\xspace}\mathbf{f}_\tplus(\mathbf{V}\xspace\cap(\mathbb{X}\times \mathbb{T}_{+}(\mathcal{L}(\mathbf{x}))\xspace|\mathbf{x}).
	\end{equation*}
	Hence, $\mathbf{f}_\tplus(\mathbf{V}\xspace|\mathbf{X}\xspace)$ can be written as
	\begin{align}
	\mathbf{f}_\tplus(\mathbf{V}\xspace|\mathbf{X}\xspace)=&1_{\mathbb{T}_{+}(\mathcal{L}(\mathbf{X)})}(\mathcal{L}(\mathbf{V})) \nonumber
	\\
	&\hspace{5pt}\times\prod_{\mathbf{x}\in\mathbf{X}\xspace}\mathbf{f}_\tplus(\mathbf{V}\xspace\cap (\mathbb{X}\times \mathbb{T}_{+}(\mathcal{L}(\mathbf{x)})\xspace|\mathbf{x}).
	\end{align}
	Substituting \eqref{eq:spawnKern} into the above equation and noting that when $1_{\mathbb{T}_{+}(\mathcal{L}(\mathbf{X)})}(\mathcal{L}(\mathbf{V}))=1$, 
	\begin{eqnarray*}
	1_{\mathbb{T}_{+}(\mathcal{L}(\mathbf{x)})}(\mathcal{L}(\mathbf{V}\xspace\cap (\mathbb{X}\times \mathbb{T}_{+}(\mathcal{L}(\mathbf{x)}))) &=&1, 
	\\
	\prod_{\mathbf{x}\in \mathbf{X}\xspace}\Delta (\mathbf{V}\xspace\cap (\mathbb{X}\times \mathbb{T}_{+}(\mathcal{L}(\mathbf{x)})) &=&\Delta (\mathbf{V}\xspace),
	\end{eqnarray*}
	we have
	\begin{equation}
	\mathbf{f}_\tplus(\mathbf{V}\xspace|\mathbf{X}\xspace)=\Delta (\mathbf{V}\xspace)1_{\mathbb{T}_{+}(\mathcal{L}(\mathbf{X)})}(\mathcal{L}(\mathbf{V}))\left[ \Phi_\tplus\left( \mathbf{V}\xspace|\mathbf{\cdot }\right) \right] ^{\mathbf{X}\xspace}
	\end{equation}
	where
	\begin{equation}
	\hspace*{-5pt}\Phi _\tplus\!\left( \mathbf{V}\xspace|\mathbf{x}\right) \!=\!\left[\Phi _\tplus(\mathbf{V}\xspace\!\cap\! (\mathbb{X}\times \mathbb{T}_{+}(\mathcal{L}(\mathbf{x}))|\mathbf{x;\cdot })\right] ^{\mathbb{T}_{+}(\mathcal{L}(\mathbf{x)})}.
	\end{equation}
	
	\begin{figure}[t]
	\begin{center}
	\begin{tikzpicture}[scale=0.65,
	rect/.style= { to path={
			let \n1={sqrt(3)*#1},
			\p1=($(\tikztostart)!\n1!-135:(\tikztotarget)$), \p2=($(\tikztostart)!\n1!135:(\tikztotarget)$),
			\p3=($(\tikztotarget)!\n1!-135:(\tikztostart)$), \p4=($(\tikztotarget)!\n1!135:(\tikztostart)$)
			in
			(\p1) -- (\p2) -- (\p3) -- (\p4) --cycle}
	},
	staterect/.style= { to path={
			let \n1={sqrt(1.5)*#1},
			\p1=($(\tikztostart)!\n1!-135:(\tikztotarget)$), \p2=($(\tikztostart)!\n1!135:(\tikztotarget)$),
			\p3=($(\tikztotarget)!\n1!-135:(\tikztostart)$), \p4=($(\tikztotarget)!\n1!135:(\tikztostart)$)
			in
			(\p1) -- (\p2) -- (\p3) -- (\p4) --cycle}
	},
	cross/.style={cross out, draw=black, fill=none, minimum size=4, inner sep=0pt, outer sep=0pt},
	cross/.default={2pt},
	]
	
	\colorlet{LightGray}{gray!40}
	\colorlet{LightRed}{red!30}
	\def\lthick{0.6pt}
	
	\coordinate (p11) at (1,2.5);
	\coordinate (p12) at (2,3.00);
	\coordinate (p13) at (3,3.50);
	\coordinate (p14) at (4,3.75);
	\coordinate (p15) at (5,4.00);
	\coordinate (p16) at (6,4.25);
	\coordinate (p17) at (7,4.50);
	\coordinate (p18) at (7.2,4.55); 
	
	\coordinate (p21) at (1,1.7);
	\coordinate (p22) at (2,1.5);
	\coordinate (p23) at (3,1.3);
	\coordinate (p24) at (4,0.8);
	\coordinate (p25) at (5,0.7);
	\coordinate (p26) at (5.2,0.68); 
	
	\coordinate (p31) at (5,3.0);
	\coordinate (p32) at (6,2.8);
	\coordinate (p33) at (7,2.6);
	\coordinate (p34) at (8,2.2);
	\coordinate (p35) at (9,1.8);
	\coordinate (p36) at (9.2,1.72);	
	
	\draw[color=Apricot, rounded corners=3.0mm, fill=Apricot] (p15) to[rect=.3cm] (p31);
	\draw[color=Apricot, rounded corners=3.0mm, fill=Apricot] (p31) to[rect=.3cm] (p32);
	\draw[color=Apricot, rounded corners=3.0mm, fill=Apricot] (p32) to[rect=.3cm] (p33);
	\draw[color=Apricot, rounded corners=3.0mm, fill=Apricot] (p33) to[rect=.3cm] (p34);
	\draw[color=Apricot, rounded corners=3.0mm, fill=Apricot] (p34) to[rect=.3cm] (p35);
	\draw[color=Apricot, rounded corners=3.0mm, fill=Apricot] (p35) to[rect=.3cm] (p36);
	
	\draw[color=LightGray, rounded corners=3.0mm, fill=LightGray] (p11) to[rect=.3cm] (p13);
	\draw[color=LightGray, rounded corners=3.0mm, fill=LightGray] (p13) to[rect=.3cm] (p17);
	\draw[color=LightGray, rounded corners=3.0mm, fill=LightGray] (p17) to[rect=.3cm] (p18);
	
	\draw[color=LightRed, rounded corners=3.0mm, fill=LightRed] (p21) to[rect=.3cm] (p23);
	\draw[color=LightRed, rounded corners=3.0mm, fill=LightRed] (p23) to[rect=.3cm] (p24);
	\draw[color=LightRed, rounded corners=3.0mm, fill=LightRed] (p24) to[rect=.3cm] (p25);
	\draw[color=LightRed, rounded corners=3.0mm, fill=LightRed] (p25) to[rect=.3cm] (p26);
	
	\draw[color=gray,line width=1pt] (0,0) rectangle (10,5);
	\foreach \x in {1,...,10}
	\draw[color=gray,line width=\lthick] (\x, 5) -- (\x, 0);
	
	\foreach \x in {0,...,5}
	\node at (\x, -0.4) {\x};
	\node at (6, -0.4) {\ldots};
	
	\draw[densely dashed, thick, fill=CornflowerBlue,fill opacity=0.3] (p11) to[staterect=.3cm] (p21);
	
	\draw[densely dashed, thick, fill=CornflowerBlue,fill opacity=0.3] (p12) to[staterect=.3cm] (p22);
	
	\draw[densely dashed, thick, fill=CornflowerBlue,fill opacity=0.3] (p13) to[staterect=.3cm] (p23);
	
	\draw[densely dashed, thick, fill=CornflowerBlue,fill opacity=0.3] (p14) to[staterect=.3cm] (p24);
	
	\draw[densely dashed, thick, fill=CornflowerBlue,fill opacity=0.3] (p15) to[staterect=.3cm] (p25);
	
	\draw[densely dashed, thick, fill=CornflowerBlue,fill opacity=0.3] (p16) to[staterect=.3cm] (p32);
	
	\draw[densely dashed, thick, fill=CornflowerBlue,fill opacity=0.3] (p17) to[staterect=.3cm] (p33);
	
	\draw[densely dashed, thick, fill=CornflowerBlue,fill opacity=0.3] (p34)+(0,0.01) to[staterect=.3cm] (p34);
	
	\draw[densely dashed, thick, fill=CornflowerBlue,fill opacity=0.3] (p35)+(0,0.01) to[staterect=.3cm] (p35);
	
	\foreach \point [count=\i] in {(p11),(p12),(p13),(p14),(p15),(p16),(p17)} {
		\draw[color=black,line width=1.5*\lthick,fill=white] \point circle (0.17);
		\draw \point node[cross,line width=1.5*\lthick] {};
	}
	
	\foreach \point [count=\i] in {(p21),(p22),(p23),(p24),(p25)} {
		\draw[color=black,line width=1.5*\lthick,fill=white] \point circle (0.17);
		\draw \point node[cross,line width=1.5*\lthick] {};
	}
	
	\foreach \point [count=\i] in {(p31),(p32),(p33),(p34),(p35)} {
		\draw[color=black,line width=1.5*\lthick,fill=white] \point circle (0.17);
		\draw \point node[cross,line width=1.5*\lthick] {};
	}
	
	\draw[densely dashed, line width=1] (p11)+(0,0.3) -- (1,5.5);
	\draw[densely dashed, line width=1] (p12)+(0,0.3) -- (2,5.5);
	\draw[densely dashed, line width=1] (p13)+(0,0.3) -- (3,5.5);
	\draw[densely dashed, line width=1] (p14)+(0,0.3) -- (4,5.5);
	\draw[densely dashed, line width=1] (p15)+(0,0.3) -- (5,5.5);
	\draw[densely dashed, line width=1] (p16)+(0,0.3) -- (6,5.5);
	\draw[densely dashed, line width=1] (p17)+(0,0.3) -- (7,5.5);
	\draw[densely dashed, line width=1] (p34)+(0,0.3) -- (8,5.5);
	\draw[densely dashed, line width=1] (p35)+(0,0.3) -- (9,5.5);
	\draw[densely dashed, line width=1] (1,5.5) -- (9,5.5);
	
	\node at (5,6) {multi-target states};
	\node [rotate=90] at (-0.333,2.5) { state space};
	\node at (5,-0.85) { time};
	
	\draw[densely dashed, thick, fill=CornflowerBlue,fill opacity=0.3] (10.2,0.3) rectangle (12.7,4.9);
	
	\node (A) at (11.4,4.5) {\small (1,2)};
	\draw[->,line width=1] (A.west) -- (7.54,4.5);
	\node (B) at (11.4,1.65) {\small (1,2,5,1)};
	\draw[->,line width=1] (B.west) -- (9.54,1.65);
	\node (C) at (11.4,0.65) {\small (1,1)};
	\draw[->,line width=1] (C.west) -- (5.54,0.65);
	\node at (11.4,5.2) { tracks};
	\end{tikzpicture}
	\end{center}
	\caption{An example of label assignment for birth and spawn tracks. Two
	tracks are born at time $1$ and are assigned labels $(1,1)$ and $(1,2)$. At
	time $5$, a track is spawned from track $(1,2)$ and is assigned label $%
	(1,2,5,1)$.}
	\label{fig:spawnlabels}
	\end{figure}
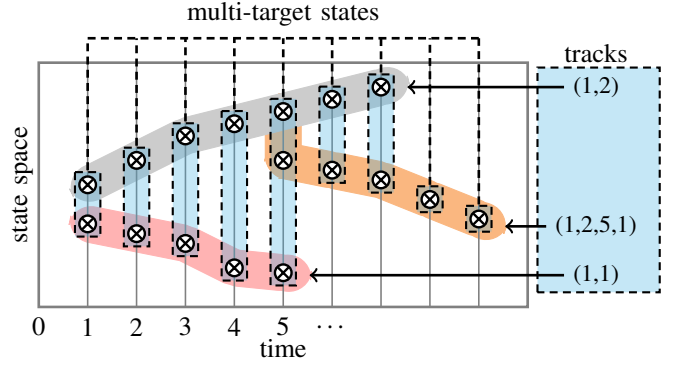

	The multi-object state at the next time $\mathbf{X}_{+}\xspace=\mathbf{X}_\splus\cup \mathbf{X}_\tplus\cup \mathbf{X}_\bplus$ is the superposition of surviving objects $\mathbf{X}_\splus=\mathbf{X}\xspace_{+}\cap (\mathbb{X}\xspace\times \mathbb{L}\xspace)$, birth objects $\mathbf{X}_\bplus=\mathbf{X}\xspace_{+}\cap (\mathbb{X}\xspace\times \mathbb{B}_{+}\xspace)$ and spawned objects $\mathbf{X}_\tplus=\mathbf{X}\xspace_{+}\cap (\mathbb{X}\xspace\times \mathbb{T}_{+}\xspace)$. Since the label spaces $\mathbb{L}$\xspace, $\mathbb{T}_{+}$\xspace, and $\mathbb{B}_{+}$\xspace are mutually disjoint, it follows that $\mathbf{X}_\splus$, $\mathbf{X}_\tplus$, and $\mathbf{X}_\bplus$ are also mutually disjoint. Further, using the conditional independence of $\mathbf{X}_\splus$, $\mathbf{X}_\tplus$, and $\mathbf{X}_\bplus$, it follows from the \ac{fisst} fundamental convolution theorem \cite{mahler2007statistical}, \cite{Vo_dGLMB_ConjPrior_2013} that the multi-object transition kernel is given by 
	\begin{equation}
	\mathbf{f}_{+}\xspace(\mathbf{X}\xspace_{+}|\mathbf{X}\xspace)\!=\!\mathbf{f}_{\S,\!+}(\mathbf{X}_\splus|\mathbf{X}\xspace)\mathbf{f}_\tplus(\mathbf{X}_\tplus|\mathbf{X}\xspace)\mathbf{f}_\bplus(\mathbf{X}_\bplus).  \label{eq:stbKern}
	\end{equation}

	\subsection{Multi-object prediction with spawning\label{sec:dglmbPred}}

	In general, for a multi-object transition density with spawning \eqref{eq:stbKern}, the GLMB family is not necessarily closed under the Chapman-Kolmogorov equation
	\begin{equation} \label{eq:ChapKolm}
	\boldsymbol{\pi }(\mathbf{X}_{+})=\int \mathbf{f}_{+}(\mathbf{X}_{+}|\mathbf{%
	X})\boldsymbol{\pi }(\mathbf{X)}\delta \mathbf{X}.
	\end{equation}

	\begin{proposition}
	If the current multi-object filtering density is \ac{glmb} of the form \eqref{eq:dglmbDist}, then the multi-object prediction density formed by surviving, birth and spawning processes is given by 
	\begin{equation}
	\boldsymbol{\pi }\xspace(\mathbf{X}\xspace_{+})=\Delta (\mathbf{X}\xspace_{+})\sum_{I,\xi }w_{+}^{(I,\xi )}(\mathcal{L}(\mathbf{X}\xspace_{+}))p^{(I,\xi)}(\mathbf{X}\xspace_{+})  \label{eq:prop1Dist}
	\end{equation}
	where
	\begin{align}
	& w_{+}^{(I,\xi )}(\mathcal{L}(\mathbf{X}\xspace_{+}))  \notag \\
	& \hspace{5pt}=w^{(I,\xi )}1_{I}(\mathcal{L}(\mathbf{X}_\splus))1_{\mathbb{T}_{+}(\mathcal{L}(\mathbf{X)})}(\mathcal{L}(\mathbf{X}_\tplus))1_{\mathbb{B}_{+}\xspace}(\mathcal{L}(\mathbf{X}_\bplus))  \notag \\
	& \hspace{10pt}\times \lbrack 1-r_\bplus]^{\mathbb{B}_{+}\xspace-\mathcal{L}(\mathbf{X}_\bplus)}[r_\bplus]^{\mathcal{L}(\mathbf{X}_\bplus)}, \\
	& p^{(I,\xi )}(\mathbf{X}\xspace_{+})  \notag \\
	& \hspace{5pt}=\left[ p_\bplus\right] ^{\mathbf{X}_\bplus}\nonumber
	\\
	&\hspace{10pt}\times\prod_{\ell \in I}\left\langle \Phi _\splus(\mathbf{X}_\splus|\cdot ,\ell )\Phi_\tplus(\mathbf{X}_\tplus|\cdot ,\ell ),p^{(\xi)}(\cdot ,\ell)\right\rangle ,  \label{eq:pIEofX+}
	\end{align}
	$\mathbf{X}_\splus=\mathbf{X}\xspace_{+}\cap (\mathbb{X}\xspace\times \mathbb{L}\xspace)$, $\mathbf{X}_\tplus=\mathbf{X}\xspace_{+}\cap(\mathbb{X}\xspace\times \mathbb{T}_{+}\xspace)$, and $\mathbf{X}_\bplus=\mathbf{X}\xspace_{+}\cap (\mathbb{X}\xspace\times \mathbb{B}_{+}\xspace)$.
	\end{proposition}

	\begin{proof}
	Using the Chapman-Kolmogorov\xspace equation \eqref{eq:ChapKolm} and \eqref{eq:stbKern} with $\mathbf{f}_\splus(\mathbf{X}_\splus\xspace|\mathbf{X}\xspace)$ from \eqref{eq:survKern} and $\mathbf{f}_\tplus(\mathbf{X}_\tplus|\mathbf{X}\xspace)$ from \eqref{eq:spawnKern}, we have 
	\begin{align}
	& \boldsymbol{\pi }\xspace(\mathbf{X}\xspace_{+})  \notag \\
	& =\mathbf{f}_\bplus(\mathbf{X}_\bplus)\!\int \!\Delta (\mathbf{X}_\splus)1_{\mathcal{L}(\mathbf{X}\xspace)}(\mathcal{L}(\mathbf{X}_\splus))\left[\Phi_\splus(\mathbf{X}_\splus|\cdot )\right]^{\mathbf{X}}\xspace  \notag \\
	& \hspace{0.5cm}\times \Delta (\mathbf{X}_\tplus)1_{\mathbb{T}_{+}(\mathcal{L}(\mathbf{X)})}(\mathcal{L}(\mathbf{X}_\tplus\xspace))\left[ \Phi_\tplus(\mathbf{X}_\tplus|\cdot )\right] ^{\mathbf{X}}\xspace  \notag \\
	& \hspace{0.5cm}\times \Delta (\mathbf{X}\xspace)\sum_{I,\xi }w^{(I,\xi)}\delta _{I}(\mathcal{L}(\mathbf{X}\xspace))\left[ p^{(\xi )}\right] ^{\mathbf{X}}\xspace\delta \mathbf{X}\xspace \\
	& \!=\Delta (\mathbf{X}_\splus)\Delta (\mathbf{X}_\tplus\xspace)\mathbf{f}_\bplus(\mathbf{X}_\bplus)\sum_{I,\xi }\!\!\int \!\!\Delta(\mathbf{X}\xspace)w^{(I,\xi )}\delta _{I}(\mathcal{L}(\mathbf{X}\xspace))  \notag \\
	& \hspace{0.5cm}\times 1_{\mathcal{L}(\mathbf{X}\xspace)}(\mathcal{L}(\mathbf{X}_\splus\xspace))1_{\mathbb{T}_{+}(\mathcal{L}(\mathbf{X)})}(\mathcal{L}(\mathbf{X}_\tplus))  \notag \\
	& \hspace{0.5cm}\times \left[ \Phi _\splus(\mathbf{X}_\splus\xspace|\cdot )\Phi _\tplus(\mathbf{X}_\tplus\xspace|\cdot )p^{(\xi )}\right] ^{\mathbf{X}}\xspace\delta \mathbf{X}\xspace \\
	& =\Delta (\mathbf{X}_\splus)\Delta (\mathbf{X}_\tplus\xspace)\mathbf{f}_\bplus(\mathbf{X}_\bplus)\sum_{I,\xi }\sum_{J\in\mathcal{F}(\mathbb{L}\xspace)}w^{(I,\xi )}\delta _{I}(J)  \notag \\
	& \hspace{0.5cm}\times 1_{J}(\mathcal{L}(\mathbf{X}_\splus))1_{\mathbb{T}_{+}(\mathcal{L}(\mathbf{X)})}(\mathcal{L}(\mathbf{X}_\tplus\xspace)) \notag \\
	& \hspace{0.5cm}\times \prod_{\ell \in I}\left\langle \Phi _\splus\xspace(\mathbf{X}_\splus|\cdot ,\ell )\Phi _\tplus\xspace(\mathbf{X}_\tplus|\cdot ,\ell),p^{(\xi )}(\cdot ,\ell )\right\rangle  \label{eq:prop1Proof}
	\end{align}
	where the last line follows from \cite[Lemma 3]{Vo_dGLMB_ConjPrior_2013}.
	Using $\mathbf{f}_\bplus(\mathbf{X}_\bplus)$ from \eqref{eq:birthKern}, $\Delta (\mathbf{X}\xspace_{+})=\Delta (\mathbf{X}_\splus)\Delta(\mathbf{X}_\tplus\xspace)\Delta (\mathbf{X}_\bplus)$, and noting that the only non-zero inner summand occurs when $I=J$, we have \eqref{eq:prop1Dist}.
	\end{proof}
	
	Implicit in \eqref{eq:prop1Dist} is that, even though the survival and spawn \acp{rfs} are mutually disjoint due to their labels (see \eqref{eq:stbKern}), they are both conditioned on the same multi-object state $\mathbf{X}$. Further, the objects spawned by a state $\mathbf{x\in X}$\xspace and its state at the next time (if survived) are all distinct, but conditioned on $\mathbf{x}$.

	\subsection{Multi-Object Update with Spawning\label{sec:glmbUpdApprox}}
	Since the predicted multi-object density is not a GLMB, the updated multi-object density 
	\begin{equation}
	\boldsymbol{\pi }_+\xspace(\mathbf{X}\xspace_{+}|Z_{+})=\frac{\boldsymbol{\pi }\xspace(\mathbf{X}\xspace_{+})g(Z_{+}|\mathbf{X}\xspace_{+})}{\int \boldsymbol{\pi }\xspace(\mathbf{X})g(Z_{+}|\mathbf{X})\delta \mathbf{X}}.
	\label{eq:approxThis}
	\end{equation}
	with the standard multi-object likelihood, is also not a GLMB. One strategy of using the GLMB filter to track with spawnings is to approximate the multi-object prediction density \eqref{eq:prop1Dist} by a GLMB prior to performing a measurement update. In this work we propose a more prudent approach whereby the GLMB approximation is performed on the updated multi-object density to reduce information loss, albeit at the cost of increased complexity. Such approximation can be achieved by finding a GLMB that matches the multi-object filtering density in the first moment and cardinality, and we do so as follows.
	
	An arbitrary labeled multi-object density on $\mathcal{F}(\mathbb{X}\xspace \times \mathbb{L}\xspace)$ can be writen in the form \cite[Proposition 3]{Papi2015}, 
	\begin{equation}
	\boldsymbol{\pi }\xspace(\mathbf{X}\xspace)=\Delta (\mathbf{X}\xspace)\sum_{c\in \mathbb{C}}w^{(c)}(\mathcal{L}(\mathbf{X}\xspace))p^{(c)}(\mathbf{X}\xspace)  \label{eq:arbLMB}
	\end{equation}
	where $\mathbb{C}$ is a discrete index set, the weights $w^{(c)}(\cdot )$ satisfy \eqref{eq:wsumto1}, and with $n=|\mathbf{X}\xspace|$, 
	\begin{equation*}
	\!\!\int \!p^{(c)}\left( \{(x_{1},\ell _{1}),\ldots ,(x_{n},\ell_{n})\}\right) d(x_{1},\ldots ,x_{n})=1.
	\end{equation*}
	Moreover, it was shown in \cite[Proposition 3]{Papi2015}, that such a labeled multi-object density can be approximated by the GLMB 
	\begin{equation}
	\hat{\boldsymbol{\pi }\xspace}(\mathbf{X}\xspace)=\Delta (\mathbf{X}\xspace)\sum_{(c,I)\in \mathbb{C}\times \mathcal{F}(\mathbb{L}\xspace)}\delta _{I}(\mathcal{L}(\mathbf{X}\xspace))\hat{w}^{(c,I)}\left[ \hat{p}^{(c,I)}\right]^{\mathbf{X}}\xspace  \label{eq:prop3Approx}
	\end{equation}
	where 
	\begin{align}
	\hat{w}^{(c,I)}& =w^{(c)}(I), \\
	\hat{p}^{(c,I)}(x,\ell )& =1_{I}(\ell )p_{I-\{\ell \}}^{(c)}(x,\ell ), \\
	p_{\{\ell _{1},\ldots ,\ell _{n}\}}^{(c)}(x,\ell )& =  \notag \\
	& \hspace{-2cm}\int \!p^{(c)}(\{(x,\ell ),(x_{1},\ell _{1}),\ldots,(x_{n},\ell _{n})\})d(x_{1},\ldots ,x_{n}).
	\end{align}
	A salient feature of this approximation method is that both the cardinality distribution and \ac{phd} of $\boldsymbol{\pi }\xspace$ are preserved. Additionally, note that $\mathbb{C}$ can take the form of any discrete index set, including the set of indices for the Cartesian product of a collection of finites subsets of some label space and an association history space, i.e., letting $\mathbb{C}=\mathcal{F}(\mathbb{L}\xspace)\times \Xi $ is possible.
	
	The exact form of the multi-object filtering density at the next time, and its GLMB approximation as per the result above, is given below in Proposition 2.
	
	\begin{figure*}[!t]
		\normalsize
		\setcounter{mytempeqncnt}{\value{equation}}
		\setcounter{equation}{41}
		
		\begin{align}
		\hat{\boldsymbol{\pi }\xspace}_+(\mathbf{X}\xspace_{+}|Z_{+}) =\Delta (\mathbf{X}\xspace_{+})&\sum_{I,\xi ,I_{+},\theta _{+}}\delta _{I_{+}}(\mathcal{L}(\mathbf{X}\xspace_{+}))\hat{w}_{+}^{(I,\xi ,I_{+},\theta_{+})}(Z_{+}) \left[ p_\bplus\psi _{+}^{(\theta _{+})}(\cdot |Z_{+})\right] ^{\mathbf{X}_\bplus}\frac{\left[ \hat{p}_{+}^{(I,\xi ,I_{+},\theta_{+})}(\cdot|Z_{+})\right] ^{\mathbf{X}_\splus\cup \mathbf{X}_\tplus}}{\left[ \bar{p}_{+}^{(I,\xi,I_{+},\theta _{+})}(\cdot |Z_{+})\right]^{I_{+}}} \label{eq:prop2Density}
		\\
		\hat{w}_+^{(I,\xi,I_+,\theta_+)}(Z_+) &= \frac{ w_+^{(I,\xi)}(I_+)\left[\bar{p}_+^{(I,\xi,I_+,\theta_+)}(\cdot|Z_+)\right]^{I_+}}{\sum\limits_{I,\xi,I_+,\theta_+}  w_+^{(I,\xi)}(I_+)\left[\bar{p}_+^{(I,\xi,I_+,\theta_+)}(\cdot|Z_+)\right]^{I_+}} \label{eq:inProp2_1}
		\\
		\hat{p}_+^{(I,\xi,I_+,\theta_+)}(x_+,\ell_+|Z_+) &= 1_{I_+}(\ell_+)\int p_+^{(I,\xi,\theta_+)}\left(\left\{ (x_+,\ell_+),(x_{1,+},\ell_{1,+}),\ldots,(x_{n,+},\ell_{n,+})\right\}|Z_+\right)d(x_{1,+},\ldots,x_{n,+}) \label{eq:inProp2_2}
		\\
		\bar{p}_+^{(I,\xi,I_+,\theta_+)}(\ell_+|Z_+) &= 1_{\Bspace}(\ell_+)\left\langle p_\bplus(\cdot,\ell_+),\psi_+^{(\theta_+)}(\cdot|Z_+)  \right\rangle + (1-1_{\Bspace}(\ell_+))\left\langle \hat{p}_+^{(I,\xi,I_+,\theta_+)}(\cdot,\ell_+|Z_+),1 \right\rangle \label{eq:inProp2_3}
		\end{align}
		\setcounter{equation}{\value{mytempeqncnt}}
		\hrulefill
		\vspace*{4pt}
	\end{figure*}

	\begin{proposition} 
		\label{prop2} If the current filtering density is GLMB of form~\eqref{eq:dglmbDist} and given the multi-object likelihood \eqref{eq:multLikeGLMB}, then the multi-object filtering density at the next time is given by 
		\begin{align}
		&\boldsymbol{\pi }_+\xspace(\mathbf{X}\xspace_{+}|Z_{+})\nonumber
		\\
		&\hspace{5pt}\propto \Delta (\mathbf{X}\xspace_{+})\sum_{I,\xi ,\theta _{+}}\!\!\!\!w_{+}^{(I,\xi)}(\mathcal{L}(\mathbf{X}\xspace_{+}))\left[ p_\bplus\psi _{+}^{(\theta _{+})}(\cdot |Z_{+})\right]^{\mathbf{X}_\bplus}  \notag \\
		& \hspace{10pt}\times p_{+}^{(I,\xi ,\theta _{+})}(\mathbf{X}_\splus\xspace\cup \mathbf{X}_\tplus|Z_{+})
		\end{align}
		where $I\in \mathcal{F}(\mathbb{L}\xspace)$, $\xi \in \Xi $, $\theta_{+}\in \Theta _{+}(\mathcal{L}(\mathbf{X}\xspace_{+}))$, $\mathbf{X}_\bplus=\mathbf{X}\xspace_{+}\cap (\mathbb{X}\xspace\times \mathbb{B}_{+}\xspace)$, $\mathbf{X}_\splus=\mathbf{X}\xspace_{+}\cap (\mathbb{X}\xspace\times \mathbb{L}\xspace)$, $\mathbf{X}_\tplus\xspace=\mathbf{X}\xspace_{+}\cap(\mathbb{X}\xspace\times \mathbb{T}_{+}\xspace)$, and 
		\begin{align}
		&p_{+}^{(I,\xi ,\theta _{+})}(\mathbf{X}_\splus\cup \mathbf{X}_\tplus\mathbf{|}Z_{+})\nonumber
		\\
		&\hspace{5pt} =\left[ \psi_{+}^{(\theta _{+})}(\cdot |Z_{+})\right]^{\mathbf{X}_\splus\xspace\cup \mathbf{X}_\tplus}  \nonumber
		\\
		&\hspace{10pt}\times\prod_{\ell \in I}\left\langle \Phi _\splus\xspace(\mathbf{X}_\splus|\cdot ,\ell )\Phi_\tplus\xspace(\mathbf{X}_\tplus|\cdot ,\ell
		),p^{(\xi )}(\cdot ,\ell )\right\rangle . \label{eq:nonBirthDist}
		\end{align}
		
		\setcounter{equation}{45}
		
		Moreover, it can be approximated by the GLMB given by \eqref{eq:prop2Density}-\eqref{eq:inProp2_3} [see top of page] which preserves the first moment and cardinality distribution, where $I_{+}\in \mathbb{L}\xspace_{+}$.
		
	\end{proposition}
	
	\begin{proof}
		With $g(Z_{+}|\mathbf{X}\xspace_{+})$ from \eqref{eq:multLikeGLMB} and $\boldsymbol{\pi }\xspace(\mathbf{X}\xspace_{+})$ from \eqref{eq:prop1Dist}, we have \begin{align}
		&\boldsymbol{\pi }_+\xspace(\mathbf{X}\xspace_{+}|Z_{+})\nonumber
		\\
		&\hspace{5pt} \propto \boldsymbol{\pi }\xspace(\mathbf{X}\xspace_{+})g(Z_{+}|\mathbf{X}\xspace_{+}),
		\\
		&\hspace{5pt} =\Delta (\mathbf{X}\xspace_{+})\!\!\!\sum_{I,\xi ,\theta_{+}}\!\!w_{+}^{(I,\xi )}(\mathcal{L}(\mathbf{X}\xspace_{+}))p_{+}^{(I,\xi)}(\mathbf{X}\xspace_{+})\nonumber
		\\
		&\hspace{15pt}\times\left[ \psi _{+}^{(\theta _{+})}(\cdot |Z_{+})\right] ^{\mathbf{X}\xspace_{+}},  
		\\
		&\hspace{5pt}=\Delta (\mathbf{X}\xspace_{+})\!\!\!\sum_{I,\xi ,\theta_{+}}\!\!w_{+}^{(I,\xi )}(\mathcal{L}(\mathbf{X}\xspace_{+}))\left[ p_\bplus\psi _{+}^{(\theta _{+})}(\cdot |Z_{+})\right] ^{\mathbf{X}_\bplus}  \notag \\
		&\hspace{15pt} \times \prod_{\ell \in I}\left\langle \Phi _\splus\xspace(\mathbf{X}_\splus|\cdot ,\ell )\Phi _\tplus\xspace(\mathbf{X}_\tplus|\cdot ,\ell ),p^{(\xi )}(\cdot,\ell )\right\rangle \nonumber
		\\
		&\hspace{15pt} \times \left[ \psi _{+}^{(\theta _{+})}(\cdot |Z_{+})\right] ^{\mathbf{X}_\splus\xspace\cup \mathbf{X}_\tplus} \\
		&\hspace{5pt}=\Delta (\mathbf{X}\xspace_{+})\!\!\!\sum_{I,\xi ,\theta_{+}}\!\!w_{+}^{(I,\xi )}(\mathcal{L}(\mathbf{X}\xspace_{+}))\left[ p_\bplus\psi _{+}^{(\theta _{+})}(\cdot |Z_{+})\right] ^{\mathbf{X}_\bplus}  \nonumber
		\\
		&\hspace{15pt}\times p_{+}^{(I,\xi ,\theta _{+})}(\mathbf{X}_\splus\xspace\cup \mathbf{X}_\tplus|Z_{+}). \label{eq:arbJointDensityNew}
		\end{align}
		
		We now apply the \ac{glmb} approximation \eqref{eq:prop3Approx} to \eqref{eq:arbJointDensityNew}, and note that determining the product of the marginals of the single-object birth densities encapsulated in $[ p_\bplus\psi _{+}^{(\theta _{+})}(\cdot |Z_{+})] ^{\mathbf{X}_\bplus}$ is redundant. Only the single-object densities encapsulated in $p_{+}^{(I,\xi ,\theta_{+})}(\mathbf{X}_\splus\xspace\cup \mathbf{X}_\tplus|Z_{+})$ require marginalization. Hence, applying the \ac{glmb} approximation from \eqref{eq:prop3Approx} gives
		\begin{align}
		\hat{\boldsymbol{\pi }\xspace}_+(\mathbf{X}\xspace_{+}|Z_{+}) =& \hat{C} \Delta (\mathbf{X}\xspace_{+})\!\!\sum_{I,\xi,I_{+},\theta _{+}}\!\!\delta _{I_{+}}(\mathcal{L}(\mathbf{X}\xspace_{+}))\nonumber
		\\
		&\hspace{5pt}\times w_+^{(I,\xi )}(\mathcal{L}(\mathbf{X}\xspace_{+}))\left[ p_\bplus\xspace\psi_+^{(\theta_+)}(\cdot|Z_{+})\right]^{\Xbold_\bplus}\nonumber
		\\
		&\hspace{5pt}\times\left[ \hat{p}^{(I,\xi ,I_{+},\theta _{+})}(\cdot|Z_+)\right]^{\Xbold_\splus \cup\Xbold_\tplus} \label{eq:approxJointDensity}
		\end{align}
		where $I_{+}\in \mathcal{F}(\mathbb{L}\xspace_{+})$, $\hat{p}^{(I,\xi ,I_{+},\theta _{+})}(\cdot|Z_+)$ is defined in \eqref{eq:inProp2_2}, and 
		\begin{align}
		\hat{C}^{-1} &=\sum_{I,\xi ,I_{+},\theta _{+}}\!\!\int\!\!\Delta (\mathbf{X}\xspace_{+})\delta _{I_{+}}(\mathcal{L}(\mathbf{X}\xspace_{+}))w^{(I,\xi )}(\mathcal{L}(\mathbf{X}\xspace_{+})) \notag \\
		& \hspace{22pt}\times \left[ p_\bplus\xspace\psi_+^{(\theta_+)}(\cdot |Z_+)\right]^{\Xbold_\bplus}\nonumber
		\\
		&\hspace{22pt}\times\left[ \hat{p}_+^{(I,\xi ,I_{+},\theta _{+})}(\cdot|Z_+)\right] ^{\Xbold_\splus \cup\Xbold_\tplus}\delta \mathbf{X}\xspace_{+} \\
		&=\sum_{I,\xi ,I_{+},\theta_{+}}\sum_{L\subseteq \mathbb{L}\xspace_{+}}\delta _{I_{+}}(L)w_+^{(I,\xi)}(L)  \notag \\
		& \hspace{22pt}\times \prod_{\ell _{+}\in L}\bar{p}_+^{(I,\xi,I_{+},\theta _{+})}(\ell _{+}|Z_+)  \label{eq:GaussEvid_1} \\
		&=\sum_{I,\xi ,I_{+},\theta_{+}}\!\!\!\!w_+^{(I,\xi )}(I_{+})\left[\bar{p}_+^{(I,\xi ,I_{+},\theta _{+})}(\cdot|Z_+ )\right] ^{I_{+}}
		\label{eq:GaussEvid_2}
		\end{align}
		where \eqref{eq:GaussEvid_1} follows from \cite[Lemma 3]{Vo_dGLMB_ConjPrior_2013}, which simplifies in \eqref{eq:GaussEvid_2} since the only non-zero inner summand occurs when $L=I_{+}$, and $\bar{p}_+^{(I,\xi ,I_{+},\theta _{+})}(\ell _{+}|Z_+)$ is defined in \eqref{eq:inProp2_3}. Substituting \eqref{eq:GaussEvid_2} into \eqref{eq:approxJointDensity} we have \eqref{eq:prop2Density} \lbrack see top of page].
	\end{proof}


	\subsection{Efficient Implementation \label{sec:Implement}}
	
	Expanding on the material discussed in Section~\ref{sec:fastImplBack}, in this section we leverage the joint prediction and update of the fast GLMB implementation to truncate the number of components used to generate the updated multi-object density in~\eqref{eq:prop2Density}. Note that a \ac{glmb} density of form \eqref{eq:dglmbDist} at the current time can be represented by%
	\begin{equation}  \label{eq:paramFormat_prior}
	\{(I^{(h)},\xi^{(h)}, w^{(h)},p^{(h)})\}_{h=1}^{H},
	\end{equation}
	which is an enumeration of the set of density parameters $\{(w^{(I,\xi)},p^{(\xi)}):(I,\xi)\in \mathcal{F}(\mathbb{L} \xspace)\times\Xi\}$ where $w^{(h)}\triangleq w^{(I^{(h)},\xi^{(h)})}$ and $p^{(h)} \triangleq p^{(\xi^{(h)})}$. The objective is to generate a parameter set%
	\begin{equation}  \label{eq:goalSet}
	\{(I_+^{(h_+)},\xi_+^{(h_+)}, w_+^{(h_+)},p_+^{(h_+)})\}_{h_+=1}^{H_+}
	\end{equation}%
	that represents the GLMB density at the next time given by~\eqref{eq:prop2Density}. Following the development in \cite[Section III-E]{Vo_2016_fast_dGLMB}, the first step toward doing so requires drawing $H_+^{\max}$ samples from the distribution $\pi$ given as%
	\begin{equation}  \label{eq:initSample}
	\pi(I,\xi) \propto w^{(I,\xi)},
	\end{equation}%
	noting that $T_+^{(h)} \xspace$ duplicates of a distinct sample $(I^{(h)} \xspace,\xi^{(h)} \xspace)$ can be drawn. Then, we determine a set of $\tilde{T}_+^{(h)} \xspace$ candidate components of the form%
	\begin{equation}  \label{eq:candComps}
	\{(I_+^{(h,t)} \xspace,\theta_+^{(h,t)} \xspace)\}_{t=1}^{\tilde{T}_+^{(h)} \xspace}
	\end{equation}%
	from each (distinct) $(I^{(h)} \xspace,\xi^{(h)} \xspace)$ that together yield significant weights, as in $\hat{w}_+^{(I,\xi,I_+,\theta_+)}(Z_+)$ in ~\eqref{eq:prop2Density}.

	Recall from Section~\ref{sec:fastImplBack} that a set of positive 1-1 vectors $\gamma$ specifies a significant weight $\hat{w}_+^{(I,\xi,I_+,\theta_+)}(Z_+)$, if $\gamma$ generates a significant $\omega^{(\gamma)} = \prod_{i=1}^P\eta_i(\gamma_i)$. Hence, determining a set of candidate components amounts to finding sets of $\gamma$'s that yield $\omega^{(\gamma)}$ above a given threshold; it follows that such vector sets can be generated using~\eqref{eq:prop2Density}. However, complexity of the GLMB filter is naturally increased with the inclusion of spawn modeling and the involvement of marginalization in the truncation procedure is inefficient, especially in cases where many GLMB components are ultimately discarded. Therefore, we exploit the following proposal density for the purpose of generating candidate components~\eqref{eq:candComps} in an effort to offset complexity and minimize inefficiency.%
	\begin{definition} 
	Given a \ac{glmb} density \eqref{eq:dglmbDist} at the current time, let the proposal density $\tilde{\boldsymbol{\pi}}$ at the next time be of form \eqref{eq:fastJointDensity} such that%
	\begin{align}  \label{eq:fastJointDensitySpawn}
	\tilde{\boldsymbol{\pi} \xspace}_+(\mathbf{X}_+|Z_+) \propto &\Delta(\mathbf{X} \xspace_+)\sum_{h=1}^{H}\sum_{t=1}^{\tilde{T}_+^{(h)} \xspace}w^{(h)} \xspace\tilde{w}_+^{(h,t)}(Z_+) \xspace  \notag \\
	&\times\delta_{I_+^{(h)} \xspace}(\mathcal{L}(\mathbf{X} \xspace_+))\left[\tilde{p}^{(h,t)}(Z_+) \xspace\right]^{\mathbf{X} \xspace_+},
	\end{align}%
	where $\psi_+^{(h,t)}\triangleq\psi_+^{(\theta_+^{(h,t)})}$, $\tilde{\psi}_+^{(h,t)}\triangleq\tilde{\psi}_+^{(\theta_+^{(h,t)})}$,
	\begin{flalign}  
	\tilde{w}_+^{(h,t)}(Z_+) \!&= \left[r_\bplus \xspace\right]^{\mathbb{B}_+ \xspace\cap I_+^{(h,t)} \xspace}\left[1-r_\bplus \xspace\right]^{\mathbb{B}_+ \xspace-I_+^{(h,t)} \xspace} \nonumber &
	\\
	&\hspace{5pt}\times\left[\bar{p}_\mathrm{S} \xspace^{(h)} \xspace\right]^{I^{(h)} \xspace\cap I_+^{(h,t)}\xspace}\left[1-\bar{p}_\mathrm{S} \xspace^{(h)} \xspace\right]^{I^{(h)} \xspace-I_+^{(h,t)} \xspace}\  \nonumber &
	\\
	&\hspace{5pt}\times\left[\bar{p}_\mathrm{T} \xspace^{(h)} \xspace\right]^{\mathbb{T}_+ \xspace\cap I_+^{(h,t)} \xspace}\left[1-\bar{p}_\mathrm{T} \xspace^{(h)} \xspace\right]^{\mathbb{T}_+ \xspace-I_+^{(h,t)} \xspace}\nonumber & 
	\\
	&\hspace{5pt}\times\left[\tilde{\psi}_+^{(h,t)}(\cdot|Z_+) \xspace\right]^{I_+^{(h,t)} \xspace},\label{eq:wTile_spawn}&
	\\
	\tilde{p}^{(h,t)} (x_+,\ell_+|Z_+\!) \!&= \frac{\tilde{p}_+^{(h)} \xspace(x_+,\ell_+)\psi^{(h,t)}(x_+,\ell_+|Z_+)}{\tilde{\psi}^{(h,t)} (x_+,\ell_+|Z_+)},&
	\\
	\tilde{p}_+^{(h)} \xspace(x_+,\ell_+) \!&=1_{\mathbb{B}_+ \xspace}(\ell_+)p_\bplus(x_+,\ell_+)  \notag &
	\\
	&\hspace{10pt}+1_\mathbb{L} \xspace(\ell_+)\tilde{p}_\mathrm{S} \xspace^{(h)}(x_+,\ell_+) \xspace \nonumber&
	\\
	&\hspace{10pt}+1_{\mathbb{T}_+ \xspace}(\ell_+)\tilde{p}_\mathrm{T} \xspace^{(h)}(x_+,\ell_+) \xspace,&
	\\
	\tilde{p}_\mathrm{S}^{(h)}\!(x_+,\!\ell_+) \!&=\!\! \frac{\langle p_\mathrm{S}\xspace(\cdot,\!\ell_+\!)f_\splus \!(x_+|\cdot,\ell),p^{(h)} \!(\cdot,\!\ell_+)\rangle}{\bar{p}_\mathrm{S} \xspace^{(h)} \xspace(\ell_+)},&
	\\
	\tilde{p}_\mathrm{T} \xspace^{(h)}(x_+,\ell_+) \!&= \frac{\langle p_\mathrm{T}\xspace(\ell_+)f_\tplus \xspace(x_+|\cdot,\ell),p^{(h)} \xspace(\cdot,\ell)\rangle}{\bar{p}_\mathrm{T} \xspace^{(h)} \xspace(\ell_+)},&
	\\
	\bar{p}_\mathrm{S} \xspace^{(h)} \xspace(\ell_+) \!&= \langle p^{(h)} \xspace(\cdot,\ell),p_\mathrm{S} \xspace(\ell_+) \rangle,&
	\\
	\bar{p}_\mathrm{T}\xspace^{(h)} \xspace(\ell_+) \!&= \langle p^{(h)} \xspace(\cdot,\ell),p_\mathrm{T} \xspace(\ell_+) \rangle,&
	\\
	\tilde{\psi}^{(h,t)} \!(x_+,\ell_+|Z_+\!) \!&=\!\langle \tilde{p}_+^{(h)} \!(x_+,\ell_+),\psi^{(h,t)}\!(x_+,\ell_+|Z_+\!)\rangle.&
	\end{flalign}
	\end{definition}%
%
	
	We enumerate $Z_+=\{z_{1:|Z_+| \xspace}\}$, $\mathbb{B}_+ \xspace=\{\ell_{1:K}\}$, $I^{(h)} \xspace=\{\ell_{K+1:L}\}$, along with the additional set of spawn labels at the next time $\Tspace(I^{(h)})=\{\ell_{L+1:P}\}$. Next we define a $P$ dimensional vector $\gamma^{(h,t)}$ that inherits the 1-1 mapping from $\theta_+^{(h,t)}$ (see~\eqref {eq:gamma}), then we recover $(I_+^{(h,t)} \xspace,\theta_+^{(h,t)} \xspace)$ via 
	\begin{subequations}
	\label{eq:g2IT}
	\begin{align}
	\hspace{-5pt}I_+^{(h,t)} \xspace&=\{\ell_i\in \mathbb{B}_+ \xspace\cup I^{(h)} \xspace\cup \Tspace(I^{(h)}):\gamma_i^{(h,t)} \xspace\geq 0\},
	\\
	\hspace{-5pt}\theta_+^{(h,t)} \xspace(\ell_i)&=\gamma_i^{(h,t)} \xspace,
	\end{align}
	such that $\theta_+^{(h,t)} \xspace:I_+^{(h,t)} \xspace\rightarrow\{0:|Z_+| \xspace\}$. Then, for each $j\in\{-1:|Z_+| \xspace\}$ define 
	\end{subequations}
	\begin{equation}  \label{eq:bstCost}
	\hspace*{-5pt}\eta_i^{(h)} \xspace(j) = 
	\begin{cases}
	1-r_\bplus \xspace(\ell_i), &\hspace{-8pt} \ell_i\!\in\!\mathbb{B}_+ \xspace, j\!<\!0, \\ 
	r_\bplus \xspace(\ell_i)\tilde{\psi}_+^{(h,j)}\!(\ell_i|Z_+\!), &\hspace{-8pt} \ell_i\!\in\!\mathbb{B}_+ \xspace, j\!\geq \!0, \\ 
	1-\bar{p}_\mathrm{S} \xspace^{(h)} \xspace(\ell_i), &\hspace{-8pt} \ell_i\!\in\! I^{(h)} \xspace, j<0, \\ 
	\bar{p}_\mathrm{S} \xspace^{(h)} \xspace(\ell_i)\tilde{\psi}_+^{(h,j)}\!(\ell_i|Z_+\!), &\hspace{-8pt} \ell_i\!\in\! I^{(h)} \xspace, j\!\geq \!0, \\ 
	1-\bar{p}_\mathrm{T} \xspace^{(h)} \xspace(\ell_i), &\hspace{-8pt} \ell_i\in\! \Tspace(I^{(h)}), j\!<\!0, \\ 
	\bar{p}_\mathrm{T} \xspace^{(h)} \xspace(\ell_i)\tilde{\psi}_+^{(h,j)}\!(\ell_i|Z_+\!), &\hspace{-8pt} \ell_i\!\in\! \Tspace(I^{(h)}), j\!\geq\! 0,
	\end{cases}
	\end{equation}
	assuming that, for all $i\in\{1:P\}$, $\bar{p}_\mathrm{S} \xspace^{(h)} \xspace(\ell_i)\in(0,1)$, $\bar{p}_\mathrm{T} \xspace^{(h)} \xspace(\ell_i)\in(0,1)$, and $\bar{p}_\dplus \xspace^{(h)} \xspace(\ell_i)\triangleq\left\langle \tilde{p}_+^{(h)} \xspace(\cdot,\ell_i),p_\dplus \xspace(\cdot,\ell_i) \right\rangle \in (0,1)$. 	We use~\eqref{eq:bstCost} in conjunction with the Gibbs sampler to yield \textit{mostly} high-weighted positive 1-1 vectors $\gamma^{(h,t)} \xspace$, then we convert each $\gamma^{(h,t)} \xspace$ to candidate component $(I_+^{(h,t)} \xspace,\theta_+^{(h,t)} \xspace)$ using~\eqref{eq:g2IT}. 
	
	Moving forward, the candidate components we determine using the proposal density~\eqref{eq:fastJointDensitySpawn}-\eqref{eq:bstCost} are subsequently used to generate the GLMB density in~\eqref{eq:prop2Density}. For each sample component $(I^{(h)},\xi^{(h)})$ and each of its candidate components~\eqref{eq:candComps} formed using~\eqref{eq:g2IT} and~\eqref{eq:bstCost}, we generate intermediate components of form  $(I_+^{(h,t)} \xspace,\xi_+^{(h,t)} \xspace,w_+^{(h,t)} \xspace,p_+^{(h,t)} \xspace)$. Letting
	\begin{equation*}
	\hat{p}_+^{(h,t)}\triangleq \hat{p}_+^{(I^{(h)},\xi^{(h)},I_+^{(h,t)},\theta_+^{(h,t)})}
	\end{equation*}
	and 
	\begin{equation*}
	\bar{p}_+^{(h,t)}\triangleq \bar{p}_+^{(I^{(h)},\xi^{(h)},I_+^{(h,t)},\theta_+^{(h,t)})},
	\end{equation*}
	we compute
	\begin{flalign}
	\bar{w}_+^{(h,t)} &= w^{(h)}\left[r_\bplus \xspace\right]^{\mathbb{B}_+ \xspace\cap I_+^{(h,t)} \xspace}\left[1-r_\bplus \xspace\right]^{\mathbb{B}_+ \xspace-I_+^{(h,t)} \xspace}\nonumber
	\\
	&\hspace{10pt}\times\left[ \bar{p}_+^{(h,t)}(\cdot|Z_+)\right]^{I_+^{(h,t)}}, \label{eq:wBlast}&
	\\
	p_+^{(h,t)}(\cdot,\ell_i) &\propto 1_{\Bspace}(\ell_i) p_\bplus(\cdot,\ell_i)\psi_+^{(h,t)}(\cdot,\ell_i|Z_+) \nonumber&
	\\
	&\hspace{15pt}+ (1-1_{\Bspace}(\ell_i))\hat{p}_+^{(h,t)}(\cdot,\ell_i|Z_+), \label{eq:p+last}&
	\end{flalign}
	define $\xi_+^{(h,t)} = (\xi^{(h)},\theta_+^{(h,t)})$, and let 
	\begin{equation}
	\hat{C} = \sum_{h,t} \bar{w}_+^{(h,t)}\left[ \bar{p}_+^{(h,t)}(\cdot|Z_+)\right]^{I_+^{(h,t)}}\label{eq:normConst}
	\end{equation} 
	be a normalizing constant. 
	Equations~\eqref{eq:wBlast}-\eqref{eq:normConst} follow directly from~\eqref{eq:prop2Density}-\eqref{eq:inProp2_3}. Algorithm~1 (see appendix) summarizes the joint prediction and update procedure for one iteration, including how the intermediate component set $\{(I_+^{(h,t)} \xspace,\xi_+^{(h,t)} \xspace,w_+^{(h,t)} \xspace,p_+^{(h,t)} \xspace)\}_{h,t=1,1}^{H,\tilde{T}_+^{(h)} \xspace}$ is formed.
	
	Note that a given set $\{(I_+^{(h,t)} \xspace,\xi_+^{(h,t)} \xspace, p_+^{(h,t)} \xspace)\}_{h,t=1,1}^{H,\tilde{T}_+^{(h)} \xspace}$ may not be unique, accordingly, the parameter set~\eqref{eq:goalSet} is determined by summing all $w_+^{(h,t)} \xspace$ with $(I_+^{(h,t)} \xspace,\xi_+^{(h,t)} \xspace) = (I_+^{(h_+)},\xi_+^{(h_+)})$, then normalizing the weights $w_+^{(h,t)} \xspace$. This procedure follows from the relationship presented in~\eqref{eq:makeDGLMB} and is summarized in Algorithm 2 , which comes from the bottom portion of \cite[Algorithm 2]{Vo_2016_fast_dGLMB} and is replicated here for convenience. All algorithms are relegated to the Appendix and follow the format of those presented in \cite{Vo_2016_fast_dGLMB} in the interest of consistency. Additionally, we use the $\mathrm{Gibbs}$ and $\mathrm{Unique}$ functions as described in \cite[Algorithm 1]{Vo_2016_fast_dGLMB} and \cite[Section III-E]{Vo_2016_fast_dGLMB}, respectively.

\section{Simulation\label{sec:simulation}}

A linear Gaussian example is used to verify the proposed \ac{glmb} filter and compare its performance to the CPHD filter; both filters incorporate object spawning. Fig.~\ref{fig:scenario} illustrates the multiple trajectories in a $[-1000,1000]\,\SI{}{\meter}\times[-1000,1000]\,\SI{}{\meter}$ surveillance region considered in this scenario. Over the $\SI{100}{\second}$ scenario duration, the number of objects varies due to birth, spawning, and death. In total, there are 6 spontaneous births and 6 spawning events. 

At the start, an object is born in each of the three birth regions. Each birth object goes on to generate a single first generation spawn, then dies. After crossing at the origin at time $k=45$, each spawn object generates a single second generation spawn. Towards the end, an object appears in each birth region that goes on to cross paths with a second generation spawn; crossings occur at times $k=82, k=84, \text{ and } k=86$ at positions $(-250,-433)$, $(-260,430)$, and $(507,26)$, respectively.

\begin{figure}[]
\centering
\includegraphics[width=\columnwidth]{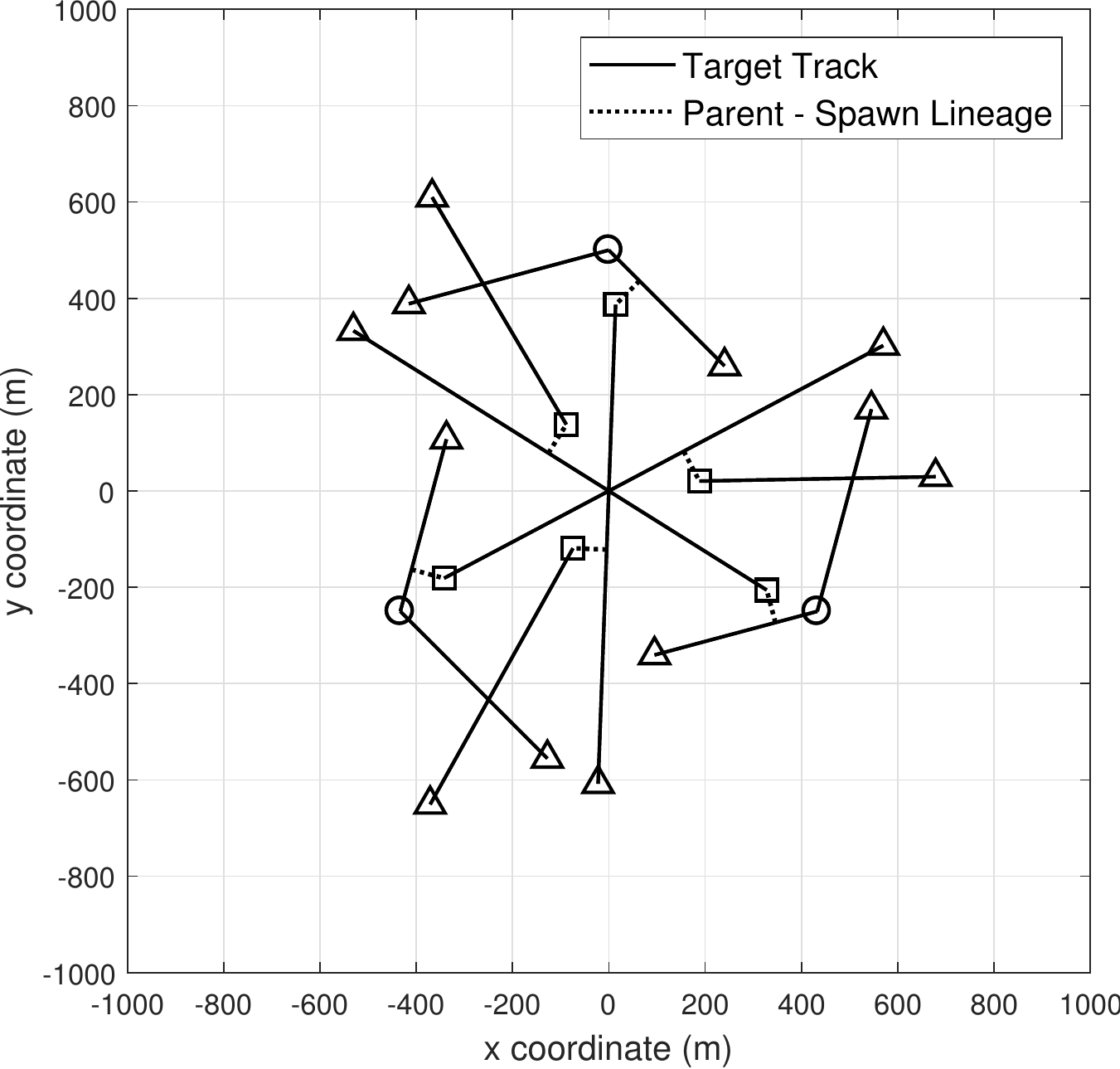}
\caption{Object trajectories in the $xy$ plane. A circle ``$\Circle $'' indicates where an object is born, a square ``$\square $'' indicates where a spawned object may be detected, and a triangle ``$\triangle $'' indicates where an object dies. }
\label{fig:scenario}
\end{figure}

The single-object state describing an object's planar position and velocity coordinates is $x_+ = \left[p_{x,+},p_{y,+},\dot{p}_{x,+},\dot{p}_{y,+}\right]^T$. Each object has a probability of survival $p_{\mathrm{S}} = 0.99$ and follows linear Gaussian dynamics with transition density $f_{\S,+}(x_+,\ell_+|x,\ell) = \mathcal{N}\left(x_+;Fx,Q\right)$ such that 
\begin{equation*}
F = 
\begin{bmatrix}
I_2 & \Delta I_2 \\ 
0_2 & I_2%
\end{bmatrix}%
,\hspace{1cm} Q = \sigma_\nu^2%
\begin{bmatrix}
\frac{\Delta^4}{4}I_2 & \frac{\Delta^3}{2}I_2 \\ 
\frac{\Delta^3}{2}I_2 & \Delta^2I_2%
\end{bmatrix}
,
\end{equation*}
where $\Delta = \SI{1}{\second}$, $\sigma_\nu = \SI{1}{\meter\second\tothe{-2}}$, and $I_n$ and $0_n$ denote the $n\times n$ identity and zero matrices, respectively.

Each object is detected with probability $p_{\mathrm{D} \xspace,+} = 0.88$ and each object generated measurement $z_+ = \left[z_{x,+},z_{y,+}\right]^T$ consists of the object's position with noise added to each component. Measurements follow the linear Gaussian measurement model $g_+(z_+|x_+,\ell_+) =  \mathcal{N}(z_+;Hx_+,R)$ such that 
\begin{equation*}
H = \left[I_2 \hspace{0.3cm} 0_2\right], \hspace{1cm} R =
\sigma_\epsilon^2I_2,
\end{equation*}
where $\sigma_\epsilon = \SI{10}{\meter}$. Clutter is modeled as a Poisson \ac{rfs} with an average intensity of $\lambda_c = 1.65\times10^{-5}\SI{}{\meter\tothe{-2}}$ yielding an average of $66$ clutter returns per scan.

Objects can appear either by birth or spawning. The birth model is an \ac{lmb} \ac{rfs} with parameters $\pi_\bplus \xspace=\{r_\bplus(\ell_i),p_\bplus (\ell_i)\}_{i=1}^3$ where $r_\bplus (\ell_i) = 0.02$ and $p_\bplus(\ell_i) = \mathcal{N}(x;m_\mathrm{B} \xspace^{(i)},P_\mathrm{B} \xspace)$ with $m_\mathrm{B} \xspace^{(1)} = \left[0,500,0,0\right]^T$, $m_\mathrm{B} \xspace^{(2)} = \left[433,-250,0,0\right]^T$, $m_\mathrm{B} \xspace^{(3)} = \left[-433,-250,0,0\right]^T$, and $P_\mathrm{B} \xspace = \sigma_\mathrm{B} \xspace^2I_4$ where $\sigma_\mathrm{B} \xspace = 10$. 

Given a parent state $\xbold=(x,\ell)$ at the current time $k$ and setting $M_\ell=1$, from~\eqref{eq:Tequation}, the set of spawn labels at the next time $k+1$ is defined as
\begin{equation*}
\Tspace(\ell) = \{(\ell,k+1)\times\{1\}\} = \{(\ell,k+1,1)\}.
\end{equation*}
Additionally, we set the probability of detection constant, i.e., $p_\T\triangleq p_\T(\xbold;\varsigma)$. Then, the spawn model is a conditional \ac{lmb} \ac{rfs} with parameters $\{(p_\T,f_\tplus(\cdot |\mathbf{x;}\varsigma)):\varsigma \in\mathbb{T}_{+}(\ell)\}$ where $p_\mathrm{T} \xspace = 0.01$ and $f_{\mathrm{T},+} \xspace(\cdot | x,\ell;\varsigma) = \sum_{i=1}^3\mathcal{N}(\cdot ;Fx + d_\mathrm{T} \xspace^{(i)},Q_\mathrm{T})$ with $Q_\mathrm{T}= \sigma_\mathrm{T} \xspace^2I_4$, $\sigma_\mathrm{T} \xspace = 5$.  Each $d_\mathrm{T} \xspace^{(i)}$ is configured such that a spawn track with zero velocity is generated at a
distance of $\SI{70}{\meter}$ from a parent state $x_k$ and in a direction relative to the parent's bearing $\theta$, i.e.,%
\begin{equation*}
d_\T^{(i)} = [70\cos(\theta + \phi^{(i)}), 70\sin(\theta + \phi^{(i)}), -\dot{p}_{x,k}, -\dot{p}_{y,k}]^T
\end{equation*}
where $\phi^{(1)}$%
$=$ $\SI{-80}{\deg}$, $\phi^{(2)} = \SI{-90}{\deg}$, and $\phi^{(3)} = \SI{-100}{\deg}$.

The maximum number of GLMB filter components is capped at $1000$. Using the Gibbs sampler to randomly generate hypotheses, the probabilities of survival and detection are tempered with values set to $\breve{p}_{\mathrm{S} \xspace,k} = 0.90p_{\mathrm{S} \xspace,k}$ and $\breve{p}_{\mathrm{D}\xspace,k} = 0.90p_{\mathrm{D} \xspace,k}$, respectively. This induces the sampler to yield more track termination and miss detection hypotheses, which expedites the termination of truly dead tracks while reducing the occurrence of dropped tracks. For more details on tempering techniques, see \cite{Vo_2016_fast_dGLMB}. The CPHD filter is configured with a Bernoulli spawn model following the presentation in \cite{Bryant_CPHD_Spawn} and is capped at 1000 Gaussian mixture components.

Results are presented for 100 Monte Carlo simulations. The mean and standard deviation of cardinality estimates over time are shown in Figs.~\ref{fig:dglmb_card} and~\ref{fig:cphd_card}, while mean \ac{ospa} \cite{schuhmacher2008consistent} distances are shown in Fig.~\ref{fig:ospa}. Mean \ac{ospa} localization and cardinality components are shown in Fig.~\ref{fig:locANDcar}.
\begin{figure}[h]
	\centering
	\includegraphics[width=1.0\columnwidth]{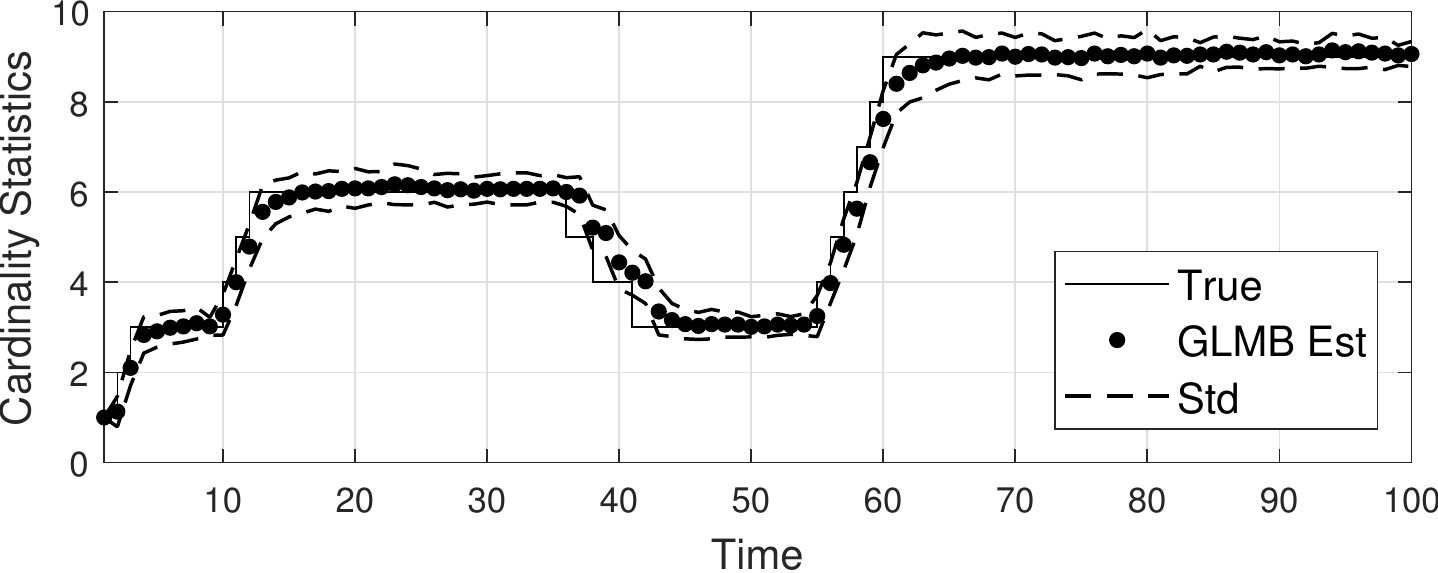}
	\caption{Cardinality statistics for GLMB filter (100 Monte Carlo trials). }
	\label{fig:dglmb_card}
\end{figure}
\begin{figure}[h]
	\centering
	\includegraphics[width=1.0\columnwidth]{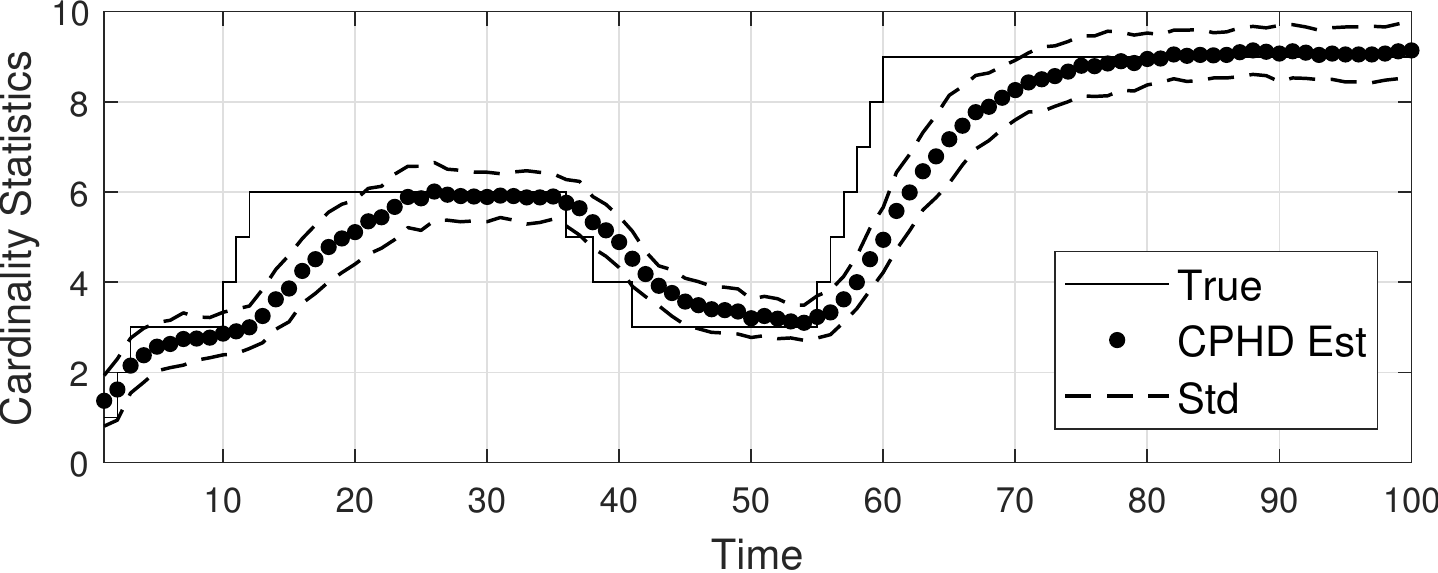}
	\caption{Cardinality statistics for CPHD filter (100 Monte Carlo trials). }
	\label{fig:cphd_card}
\end{figure}
\begin{figure}[h]
	\centering
	\includegraphics[width=1.0\columnwidth]{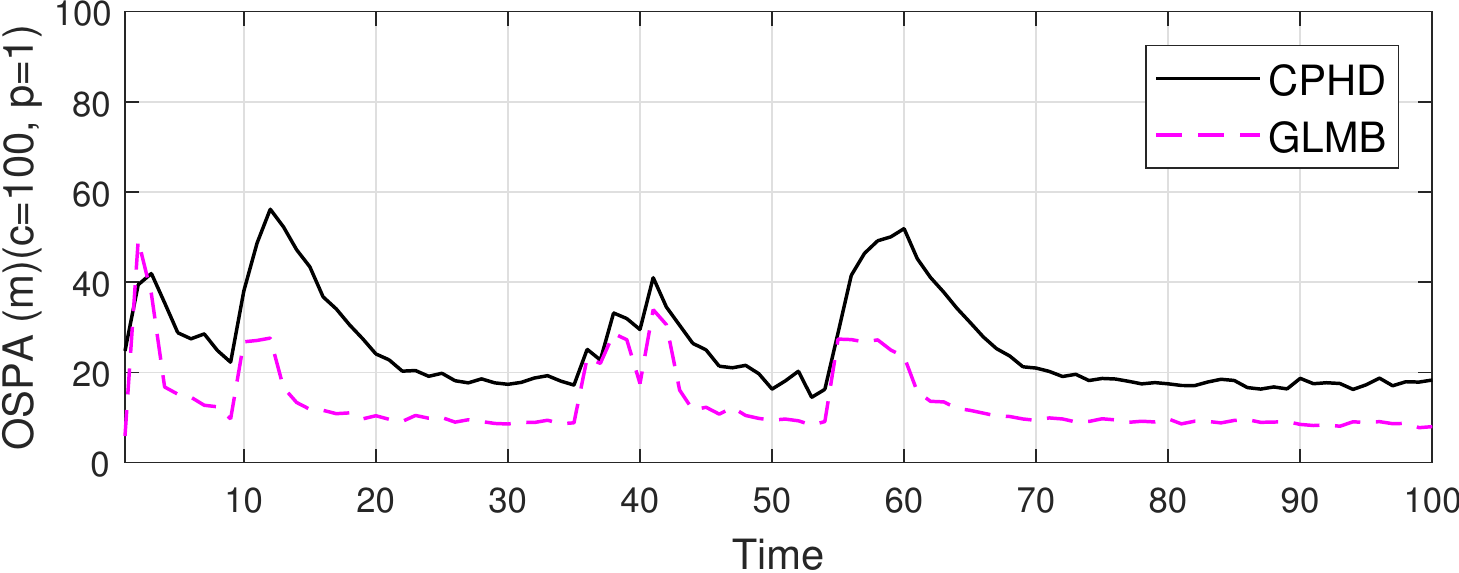}
	\caption{OSPA distance for GLMB and CPHD filters (100 Monte Carlo trials). }
	\label{fig:ospa}
\end{figure}
\begin{figure}[h]
	\centering
	\includegraphics[width=1.0\columnwidth]{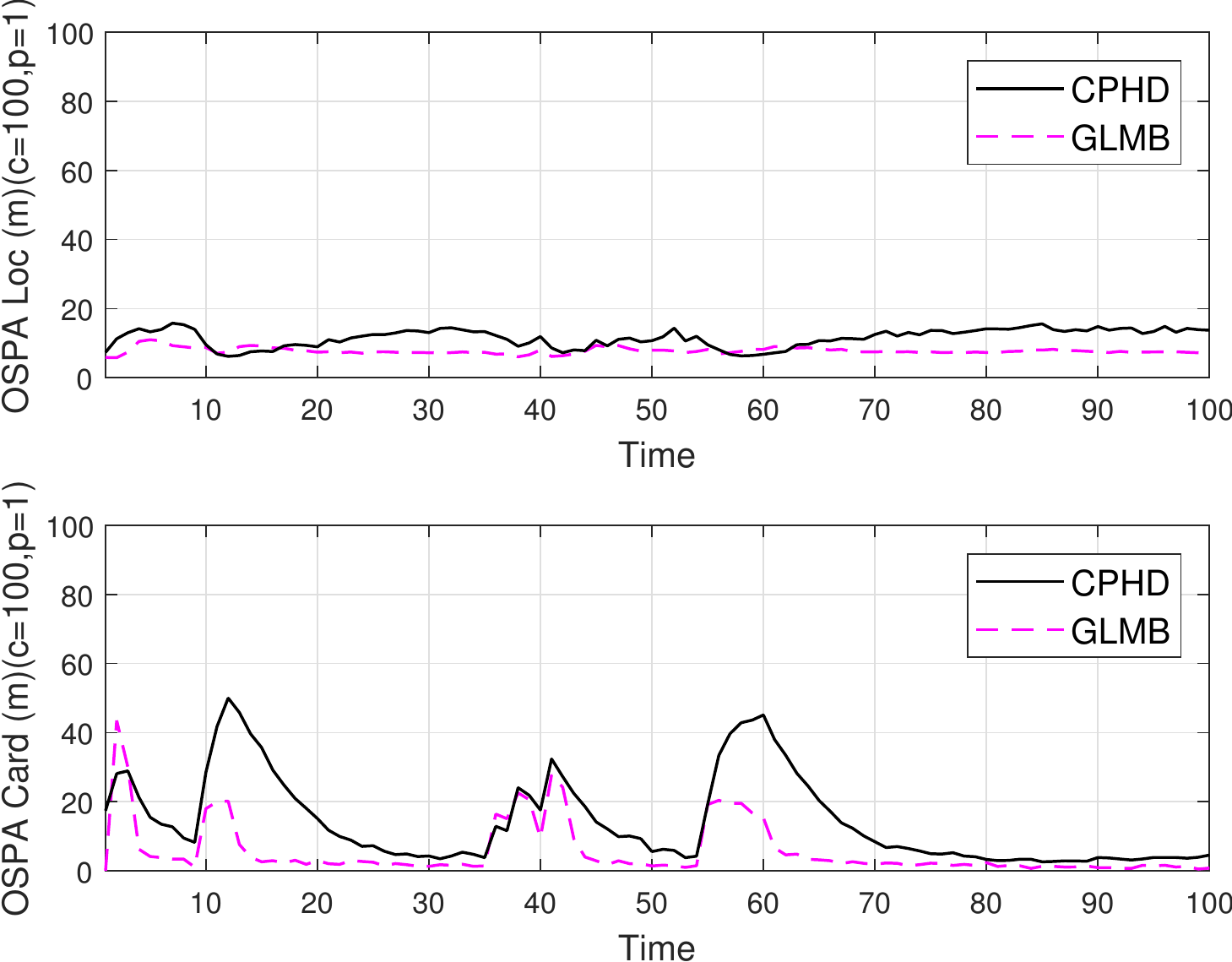}
	\caption{OSPA components for GLMB and CPHD filters (100 Monte Carlo trials). }
	\label{fig:locANDcar}
\end{figure}

Similar to the results presented in \cite{Vo_dGLMB_Sequel_2014}, both filters accurately estimate cardinality with the GLMB filter providing a better cardinality variance estimate. The GLMB also exhibits better miss distance performance throughout the majority of
the simulation. From Fig.~\ref{fig:locANDcar} we see that the GLMB outperforms the CPHD in both cardinality and localization components overall. Consistent with the assessment in \cite{Vo_dGLMB_Sequel_2014}, lower estimated cardinality variance promotes improved cardinality performance. Improved localization performance is attributed to the GLMB's ability to propagate the filtering density more accurately and its immunity to the ``spooky-effect". The CPHD filter is susceptible to this effect where in the event of a miss detection the PHD \textit{mass} shifts away from undetected components to detected ones, regardless of the distance between them \cite{5310327, Vo_dGLMB_Sequel_2014}.

The proposed GLMB filter's ability to capture ancestry information is demonstrated in Fig.~\ref{fig:wholeFig}. For each Monte Carlo run, the final label estimates at time $k=100$ are compared to the set of true  labels of the same time which are presented in Table~\ref{tab:Lend}. The true labels are presented horizontally and are enumerated by the far left column. Shading is used to divide labels 1-6 into segments correlating with birth and spawning events in Fig.~\ref{fig:wholeFig}, where red, green, and yellow represent a birth track from which spawn originate, first generation spawn, and second generation spawn, respectively. Labels 7-9 are not considered for this analysis since they have no ancestry.

When comparing the final label estimates at time $k=100$ of each Monte Carlo run with the truth, estimated labels are segregated by common ancestry, e.g., from Table~\ref{tab:Lend}, labels 1 and 4 belong to one group, labels 2 and 5 belong to another, etc. For a given label group, the time of birth for their common birth ancestor track may differ from the truth, therefore it is necessary to trace the common ancestor track's label and state estimates toward the beginning of the given Monte Carlo run. Comparing this time history of state estimates with the truth, the originating birth region is determined. Then, the birth track's time of death is found by tracing its label from the beginning of the given run to the point in time when its label is no longer present in the set of label estimates. The remaining event times represented in Fig.~\ref{fig:wholeFig}, i.e., times of birth and spawning, are extracted from the labels of a final label estimates at time $k=100$ of the given Monte Carlo run.

\def\arraystretch{1.1}
\newcolumntype{g}{>{\columncolor{red!50}}c}
\newcolumntype{e}{>{\columncolor{lightGreen}}c}
\newcolumntype{f}{>{\columncolor{yellah}}c}
\begin{table}[H]
	\caption{Label Ground Truth at Time $k=100$\label{tab:Lend}}
	\centering
	\begin{tabular}{@{}cggeeff@{}}\toprule
		\rowcolor{white}
		Label \# &&&Label&&&\\
		\cmidrule{1-7} 	
		1 & 1 & 1 & 10 & 1 & 56 & 1  \\
		2 & 2 & 2 & 11 & 1 & 58 & 1  \\
		3 & 3 & 3 & 12 & 1 & 60 & 1  \\
		4 & 1 & 1 & 10 & 1 & \cellcolor{white}   & \cellcolor{white}   \\
		5 & 2 & 2 & 11 & 1 & \cellcolor{white}   & \cellcolor{white}   \\
		6 & 3 & 3 & 12 & 1 & \cellcolor{white}   & \cellcolor{white}   \\
		\rowcolor{white}		7 & 55 & 3 &  &  &  & \\
		\rowcolor{white}		8 & 57 & 1 &  &  &  & \\
		\rowcolor{white}		9 & 59 & 2 &  &  &  & \\
		\bottomrule
	\end{tabular}
\end{table}

The formats of Figs.~\ref{fig:anc_sub1}-\ref{fig:anc_sub3} are generally the same. Each birth region's true ancestry tree is at the far left and is aligned with a gridded area to the right with markers indicating GLMB filter estimated birth, death, and spawn times for 100 Monte Carlo runs. Each figure corresponds to one of the three modeled birth regions as indicated by the true track labels.
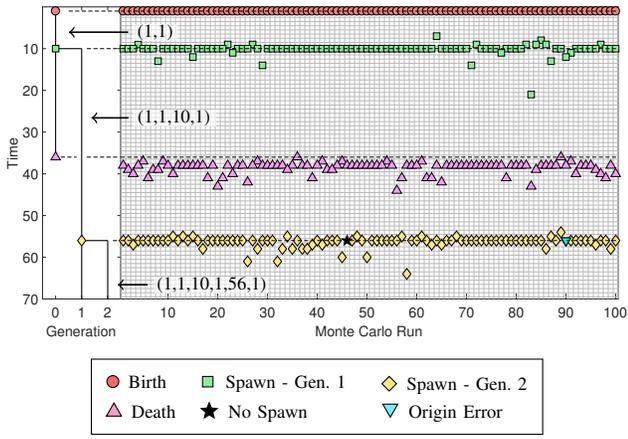
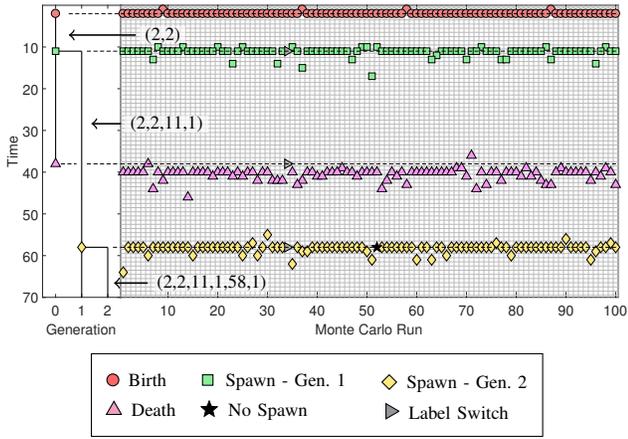
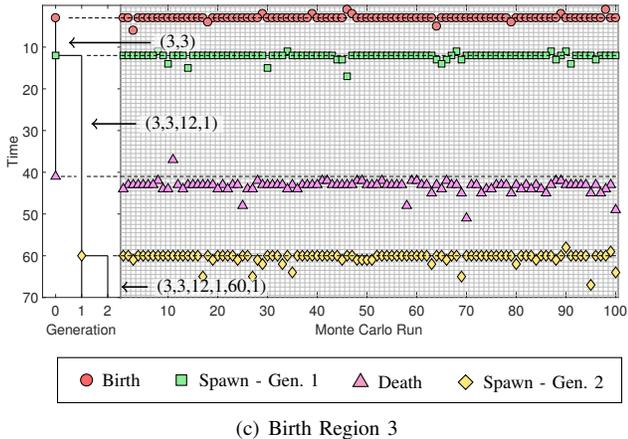
\begin{figure}[t!]
	\begin{center}
		\subfigure[Birth Region 1]{%
			\label{fig:anc_sub1}
			\begin{tikzpicture}
			\tikzstyle{legBox}=[draw, rectangle,minimum height=1cm, minimum width=6cm, fill=white]
			\tikzstyle{legEntry}=[draw, white,rectangle,fill=white,text=black]

			\node (ancestry) at (0,0)
			{\includegraphics[width = 0.93\columnwidth]{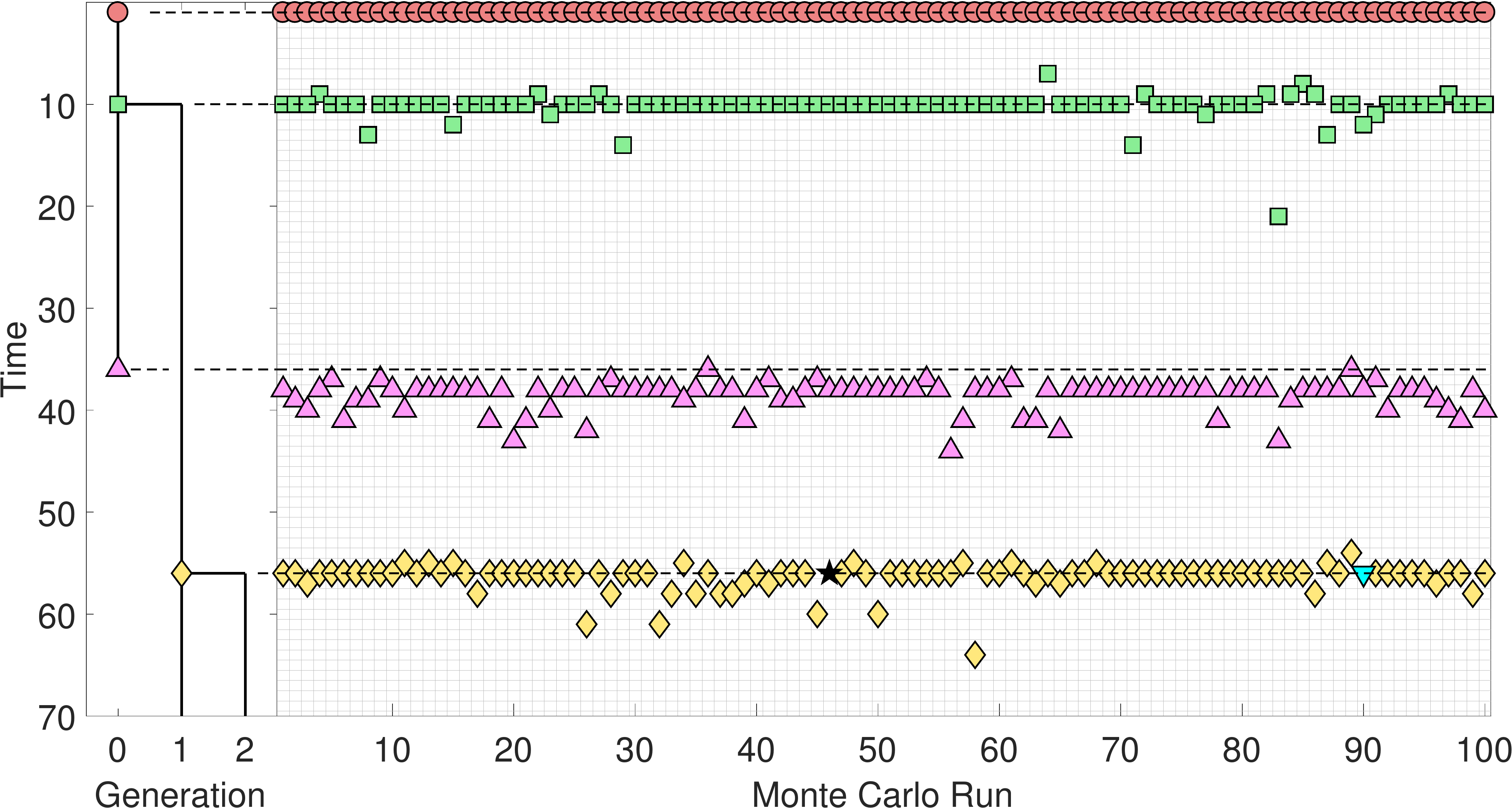}};
			
			\fill [white,opacity=0.6] (-2.45,1.74) rectangle node[black,pos=0.5,opacity=1.0] (L1) {\scriptsize (1,1)} (-1.85,1.96);
			\draw[->,line width=0.6] (L1.west) -- (-3.3,1.85);
			
			\fill [white,opacity=0.6] (-2.45,0.6) rectangle node[black,pos=0.5,opacity=1.0] (L2) {\scriptsize (1,1,10,1)} (-1.35,0.82);
			\draw[->,line width=0.6] (L2.west) -- (-3.0,0.71);
			
			\fill [white,opacity=0.6] (-2.15,-1.4) rectangle node[black,pos=0.5,opacity=1.0] (L3) {\scriptsize (1,1,10,1,56,1)} (-0.65,-1.62);
			\draw[->,line width=0.6] (L3.west) -- (-2.65,-1.51);

			\coordinate (legCenter) at (0,-3);
			\coordinate (legLeft) at (-2.4,-2.8);
			\coordinate (legLeftBottom) at (-2.7,-3.2);
			
			\node[legBox] at (legCenter) (legend) {};
			\node[legEntry] at (legLeft) (E1) {\tikz\draw[black,fill=red!60] (0,0) circle (.5ex); \scriptsize Birth};
			
			\node[right=0.35cm of E1,draw, thin, fill=lightGreen, regular polygon, regular polygon sides =4, scale=0.42] (E2) {};
			\node[legEntry,below right=-0.28cm and 0.04cm of E2] (E3) {\scriptsize Spawn - Gen. 1};
			
			\node[right=0.35cm of E3,draw, thin, fill=yellah, diamond, scale=0.42] (E4) {};
			\node[legEntry,below right=-0.26cm and 0.06cm of E4] (E5) {\scriptsize Spawn - Gen. 2};
			
			\node[draw, thin, fill=Lavender, regular polygon, regular polygon sides =3, scale=0.35] at (legLeftBottom) (E6) {};
			\node[legEntry,right=0.04cm of E6] (E7) {\scriptsize Death};
			
			\node[below right=-0.26cm and 0.28cm of E7,draw, thin, fill=black, star, star points=5, star point ratio=2.25, scale=0.3] (E8) {};
			\node[legEntry,below right=-0.26cm and 0.07cm of E8] (E9) {\scriptsize No Spawn};
			
			\node[below right=-0.26cm and 1.02cm of E9,draw, thin, fill=cblue, regular polygon, regular polygon sides =3, rotate=180, scale=0.35] (E10) {};
			\node[legEntry,below right=-0.11cm and 0.18cm of E10] (E11) {\scriptsize Origin Error};
			
			\end{tikzpicture}
		}\\%
		\subfigure[Birth Region 2]{%
			\label{fig:anc_sub2}
			\begin{tikzpicture}
			\tikzstyle{legBox}=[draw, rectangle,minimum height=1.1cm, minimum width=6cm, fill=white]
			\tikzstyle{legEntry}=[draw, white,rectangle,fill=white,text=black]

			\node (ancestry) at (0,0)
			{\includegraphics[width = 0.93\columnwidth]{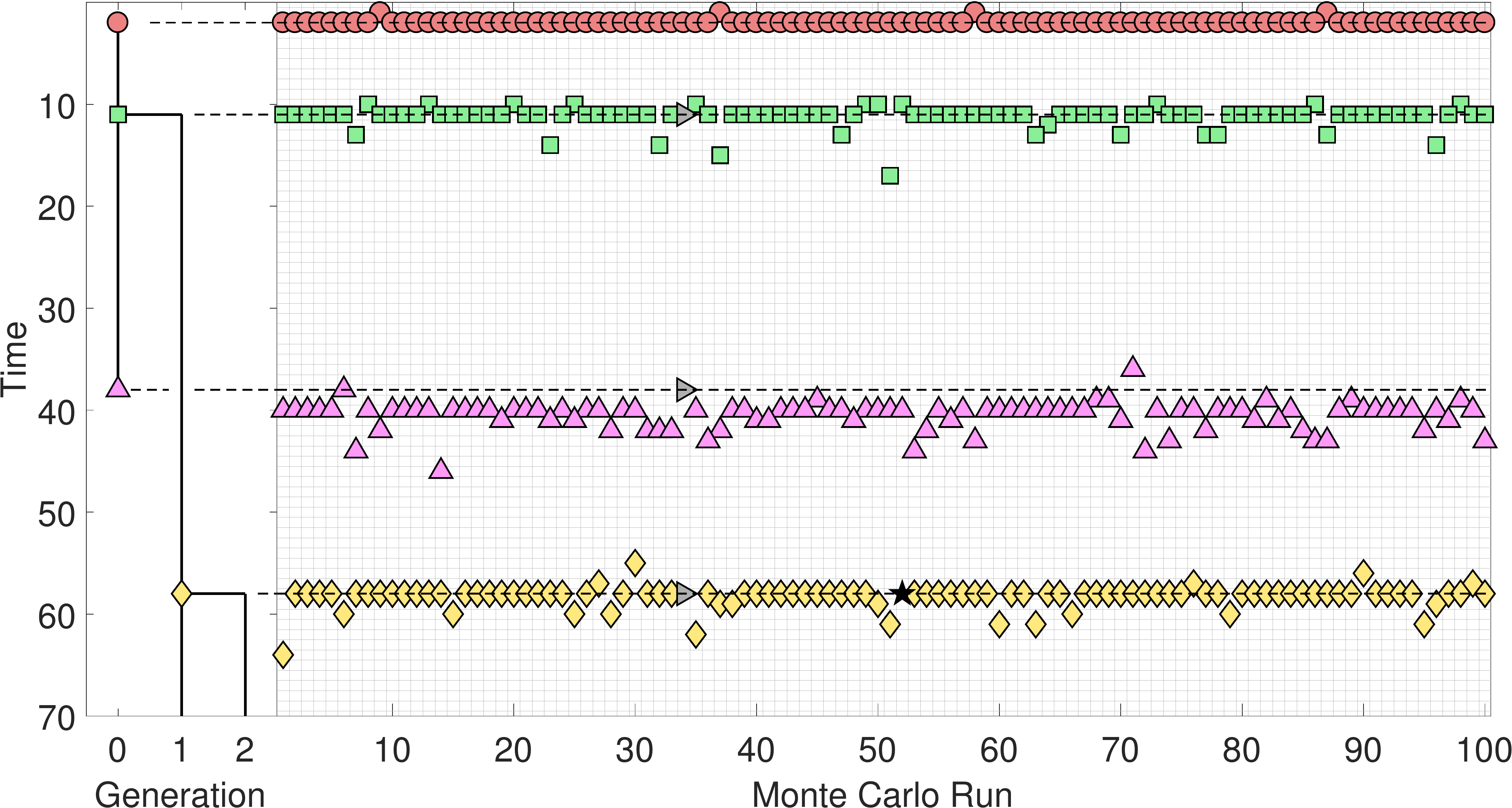}};
			
			\fill [white,opacity=0.6] (-2.35,1.68) rectangle node[black,pos=0.5,opacity=1.0] (L1) {\scriptsize (2,2)} (-1.75,1.90);
			\draw[->,line width=0.6] (L1.west) -- (-3.3,1.79);
			
			\fill [white,opacity=0.6] (-2.55,0.5) rectangle node[black,pos=0.5,opacity=1.0] (L2) {\scriptsize (2,2,11,1)} (-1.45,0.72);
			\draw[->,line width=0.6] (L2.west) -- (-3.0,0.61);
			
			\fill [white,opacity=0.6] (-2.15,-1.4) rectangle node[black,pos=0.5,opacity=1.0] (L3) {\scriptsize (2,2,11,1,58,1)} (-0.65,-1.62);
			\draw[->,line width=0.6] (L3.west) -- (-2.7,-1.51);

			\coordinate (legCenter) at (0,-3);
			\coordinate (legLeft) at (-2.4,-2.8);
			\coordinate (legLeftBottom) at (-2.7,-3.2);
			
			\node[legBox] at (legCenter) (legend) {};
			\node[legEntry] at (legLeft) (E1) {\tikz\draw[black,fill=red!60] (0,0) circle (.5ex); \scriptsize Birth};
			
			\node[right=0.35cm of E1,draw, thin, fill=lightGreen, regular polygon, regular polygon sides =4, scale=0.42] (E2) {};
			\node[legEntry,below right=-0.28cm and 0.04cm of E2] (E3) {\scriptsize Spawn - Gen. 1};
			
			\node[right=0.35cm of E3,draw, thin, fill=yellah, diamond, scale=0.42] (E4) {};
			\node[legEntry,below right=-0.26cm and 0.06cm of E4] (E5) {\scriptsize Spawn - Gen. 2};
			
			\node[draw, thin, fill=Lavender, regular polygon, regular polygon sides =3, scale=0.35] at (legLeftBottom) (E6) {};
			\node[legEntry,right=0.04cm of E6] (E7) {\scriptsize Death};
			
			\node[below right=-0.26cm and 0.28cm of E7,draw, thin, fill=black, star, star points=5, star point ratio=2.25, scale=0.3] (E8) {};
			\node[legEntry,below right=-0.26cm and 0.07cm of E8] (E9) {\scriptsize No Spawn};
			
			\node[below right=-0.29cm and 1.02cm of E9,draw, thin, fill=Gray, regular polygon, regular polygon sides =3, rotate=270, scale=0.35] (E10) {};
			\node[legEntry,below right=-0.25cm and 0.18cm of E10] (E11) {\scriptsize Label Switch};
			
			\end{tikzpicture}
		}\\%
		\subfigure[Birth Region 3]{%
			\label{fig:anc_sub3}
			\begin{tikzpicture}
			\tikzstyle{legBox}=[draw, rectangle,minimum height=0.6cm, minimum width=7.5cm, fill=white]
			\tikzstyle{legEntry}=[draw, white,rectangle,fill=white,text=black]

			\node (ancestry) at (0,0)
			{\includegraphics[width = 0.93 \columnwidth]{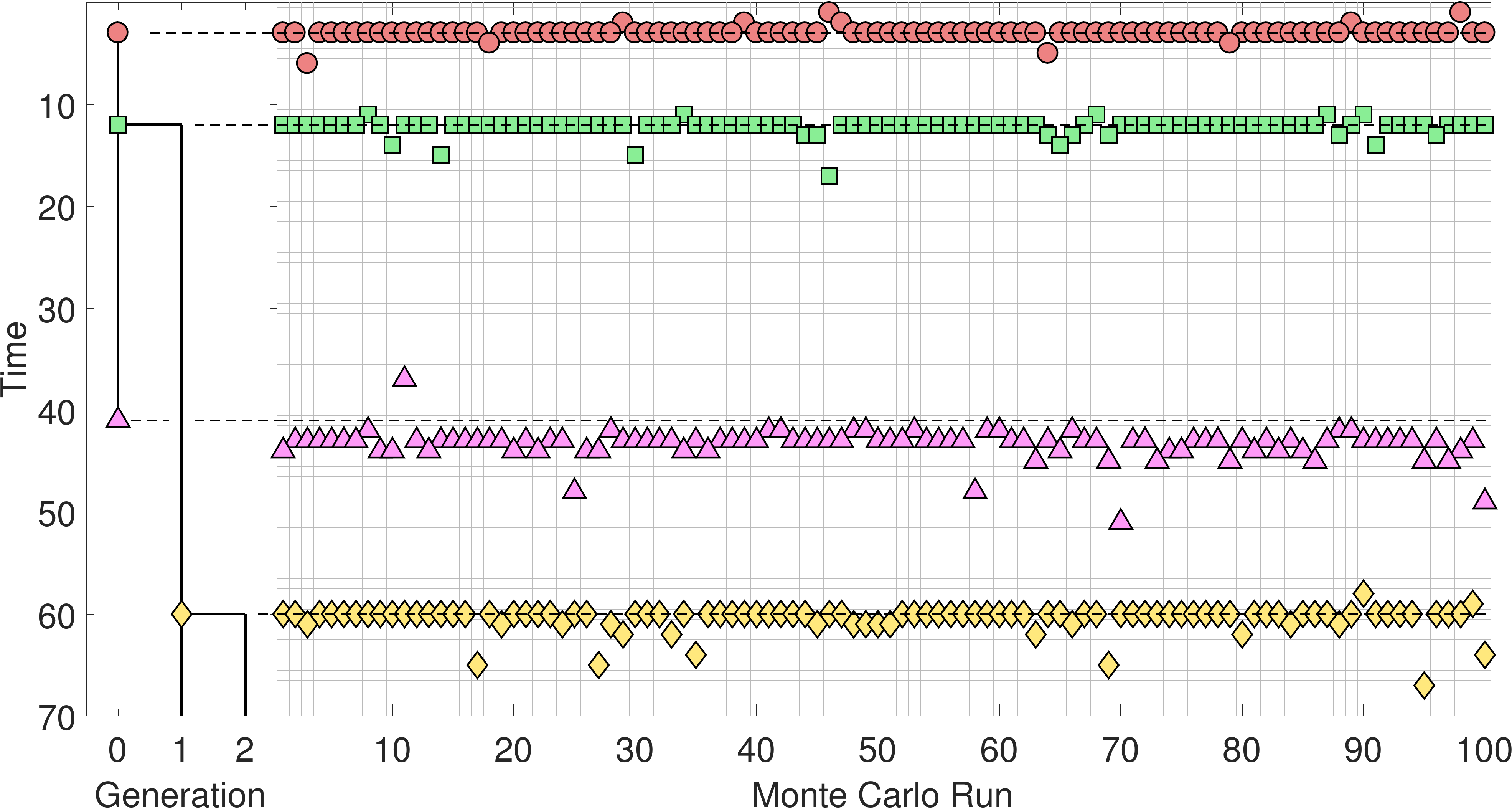}};
			
			\fill [white,opacity=0.6] (-2.15,1.58) rectangle node[black,pos=0.5,opacity=1.0] (L1) {\scriptsize (3,3)} (-1.55,1.80);
			\draw[->,line width=0.6] (L1.west) -- (-3.3,1.69);
			
			\fill [white,opacity=0.6] (-2.35,0.5) rectangle node[black,pos=0.5,opacity=1.0] (L2) {\scriptsize (3,3,12,1)} (-1.25,0.72);
			\draw[->,line width=0.6] (L2.west) -- (-3.0,0.61);
			
			\fill [white,opacity=0.6] (-2.15,-1.45) rectangle node[black,pos=0.5,opacity=1.0] (L3) {\scriptsize (3,3,12,1,60,1)} (-0.65,-1.67);
			\draw[->,line width=0.6] (L3.west) -- (-2.6,-1.56);
			
			\coordinate (legCenter) at (0.3,-2.8);
			\coordinate (legLeft) at (-2.75,-2.8);

			\node[legBox] at (legCenter) (legend) {};
			\node[legEntry] at (legLeft) (E1) {\tikz\draw[black,fill=red!60] (0,0) circle (.5ex); \scriptsize Birth};
			
			\node[right=0.35cm of E1,draw, thin, fill=lightGreen, regular polygon, regular polygon sides =4, scale=0.42] (E2) {};
			\node[legEntry,below right=-0.28cm and 0.04cm of E2] (E3) {\scriptsize Spawn - Gen. 1};
			
			\node[right=0.35cm of E3,draw, thin, fill=Lavender, regular polygon, regular polygon sides =3, scale=0.35] (E4) {};
			\node[legEntry,right=0.04cm of E4] (E5) {\scriptsize Death};
			
			\node[right=0.35cm of E5,draw, thin, fill=yellah, diamond, scale=0.42] (E6) {};
			\node[legEntry,below right=-0.26cm and 0.06cm of E6] (E7) {\scriptsize Spawn - Gen. 2};
			
			\end{tikzpicture}
		}
	\end{center}
	\caption{Birth region ancestry truth and estimates. Each region's true ancestry tree is at the far left. A red circle at the top indicates birth, while the first and second generation spawn times are denoted by a green square and yellow diamond, respectively. The birth track's time of death is marked by a pink diamond. Markers within the gridded area indicate GLMB filter estimated birth, death, and spawn times for 100 Monte Carlo runs.}
	\label{fig:wholeFig}
\end{figure}%
We see from Fig.~\ref{fig:wholeFig} that the GLMB filter accurately captures ancestry information overall, where the results in Fig.~\ref{fig:anc_sub3} exhibit the best performance. Fig.~\ref{fig:anc_sub1} indicates that the second generation spawn track $(1,1,10,1,56,1)$ was not spawned during run $46$ and that the same generation track was estimated as originating from a different parent during run 90. In this specific case, track $(1,1,10,1)$ dropped prior to tracks crossing at the origin (see Fig.~\ref{fig:scenario}) and was later estimated as having spawned from one of the remaining first generation spawn tracks, subsequently spawning an object at time $k=56$. Similarly, Fig.~\ref{fig:anc_sub2} shows that track $(2,2,11,1,58,1)$ did not spawn during run 52, while during run 34, a label switch occurred when track $(2,2)$ essentially took the place of track $(2,2,11,1)$, going on to spawn track $(2,2,58,1)$ at time $k=58$. These inaccurate ancestry estimates are due to missed detections of either the parent track, spawn track, or both.%

\section{Conclusion \label{sec:conclusion}}

This paper presented the first GLMB filter to consider object spawning. Using a top-down formulation, a general labeled \ac{rfs} density characterizing the predicted multi-object density of surviving, birth and spawn objects was derived. A joint prediction-update was performed yielding a density that was then approximated to form a posterior GLMB density while preserving its cardinality and \ac{phd}. A key innovation of the proposed filter is the capacity of spawn track labels to encapsulate their ancestry. The filter's ability to instantiate 
new tracks originating from previously known objects was verified by simulation. Our results can potentially be extended to accommodate measurement-based birth and spawn models.

\FloatBarrier
\appendix
\begin{figure}[h!]
\hrule
\vspace{3pt} \textbf{Algorithm 1: Joint Prediction and Update with Spawning}
\par
\begin{itemize}
\item input:$\left\{\left(I^{(h)},\xi^{(h)},w^{(h)},p^{(h)}\right)\right%
\}_{h=1}^{H},\; Z_+,\; H_+^{\max}$,
\par
\item input: $\{(r_\bplus \xspace^{(\ell)},\; p_\bplus^{(\ell)})\}_{\ell\in\mathbb{B}_+ \xspace},\; p_\mathrm{S} \xspace,\; f_{\mathrm{S},+}\xspace(\cdot|\cdot)$,
\item input: $p_\mathrm{T} \xspace,\; f_{\mathrm{T},+} \xspace(\cdot|\cdot),\;\Tspace(\cdot), \kappa_+,\; p_{\mathrm{D} \xspace,+},\; g_+(\cdot|\cdot)$%
\par
\item output:$\left\{\left(I^{(h_+)},\xi^{(h_+)},w^{(h_+)},p^{(h_+)}\right)%
\right\}_{h_+=1}^{H_+}$%
\end{itemize}
\par
\vspace{3pt} \hrule
\vspace{3pt} 
\begin{algorithmic}[1]
		\item sample counts $[T_+^{(h)}]_{h=1}^{H}$ from a multinomial distribution with parameters $H_+^{\max}$ trials and weights $[w^\hex]_{h=1}^H$, 
		\item for $h=1:H$
		\item \tab generate $\Tspace(I^{(h)})\!\!=\!\!\{\{(\ell,k+1)\}\!\!\times\!\!\{1\!:\!M_\ell\}\!:\!\ell\in I^{(h)}\}$
		\item \tab initialize $\gamma^{(h,1)}$
		\item \tab compute $\eta^\hex \!\!=\!\! [\eta_i^\hex(j)]_{(i,j)=(1,-1)}^{(|\Bspace\cup I^\hex\cup \Tspace(I^{(h)})|, |Z_+|)}$ 
		\item[] \tab using~\eqref{eq:bstCost}
		\item \tab $\{\gamma^\htex\}_{t=1}^{\tilde{T}_+^\hex}:= \mathrm{ Unique}(\mathrm{Gibbs}(\gamma^{(h,1)}, T_+^\hex, \eta^\hex))$
		\item \tab for $t = 1:\tilde{T}_+^\hex$
		\item \tab\tab convert $\gamma^\htex$ to $(I_+^\htex,\theta_+^\htex)$ using~\eqref{eq:g2IT}
		\item \tab\tab compute $\bar{w}_+^{(h,t)}$ from $w^{(h)}$ and $I_+^\htex$ using~\eqref{eq:wBlast}
		\item \tab\tab compute and normalize $p_+^{(h,t)}$ using~\eqref{eq:p+last}
		\item \tab\tab $\xi_+^{(h,t)}=(\xi^{(h)},\theta_+^{(h,t)})$
		\item \tab end
		\item end
		\item compute $\hat{C}$ given in~\eqref{eq:normConst}
		\item compute $w_+^{(h,t)} = \bar{w}_+^{(h,t)} /\hat{C}$
		\item $\{(I_+^{(h_+)},\xi_+^{(h_+)},w_+^{(h_+)},p_+^{(h_+)})\}_{h_+=1}^{H_+}$
		\item[] \tab $:=\mathrm{Aggregate}(\{(I_+^\htex,\xi_+^{(h,t)},w_+^{(h,t)},p_+^{(h,t)})\}_{h,t=1,1}^{H,\tilde{T}_+^\hex})$
	\end{algorithmic}
\vspace{3pt} \hrule
\end{figure}
\begin{figure}[t]
\hrule
\vspace{3pt} \textbf{Algorithm 2: Aggregate} \newline
From the bottom portion of \cite[Algorithm 2]{Vo_2016_fast_dGLMB} replicated
here for convenience, though with the addition of $\xi_+^{(h,t)} \xspace$
terms
\par
\begin{itemize}
\item input:$\{(I_+^{(h,t)} \xspace,\xi_+^{(h,t)},w_+^{(h,t)},p_+^{(h,t)})%
\}_{h,t=1,1}^{H,\tilde{T}_+^{(h)} \xspace} $%
\par
\item output: $\{(I_+^{(h_+)},\xi_+^{(h_+)},w_+^{(h_+)},p_+^{(h_+)})%
\}_{h_+=1}^{H_+}$%
\end{itemize}
\par
\vspace{3pt} \hrule
\vspace{3pt} 
\begin{algorithmic}[1]
		\item $(\{(I_+^{(h_+)},\xi_+^{(h_+)},p_+^{(h_+)})\}_{h_+=1}^{H_+},\sim,[U_{h,t}])$
		\item[] \tab\tab $:=\mathrm{Unique}(\{(I_+^\htex,\xi_+^{(h,t)},p_+^{(h,t)})\}_{h,t=1,1}^{H,\tilde{T}_+^\hex})$
		\item for $h_+ = 1:H_+$
		\item \tab $w_+^{(h_+)} = \sum\limits_{h,t:U_{h,t}=h_+}w_+^{(h,t)}$
		\item end
		\item normalize weights $\{w_+^{(h_+)}\}_{h_+=1}^{H_+}$
	\end{algorithmic}
\vspace{3pt} \hrule
\end{figure}

\FloatBarrier
\bibliographystyle{IEEEtran}
\bibliography{IEEEabrv,L:/Research/Library/mendeley/library_new} 

\end{document}